\documentclass[letterpaper]{article} %
\usepackage{aaai23}  %
\usepackage{times}  %
\usepackage{helvet}  %
\usepackage{courier}  %
\usepackage[hyphens]{url}  %
\usepackage{graphicx} %
\usepackage{natbib}  %
\usepackage{caption} %
\usepackage{algorithm}

\usepackage{newfloat}
\usepackage{listings}
\DeclareCaptionStyle{ruled}{labelfont=normalfont,labelsep=colon,strut=off} %
\floatstyle{ruled}
\newfloat{listing}{tb}{lst}{}
\floatname{listing}{Listing}
\usepackage{pst-node}
\usepackage{tikz-cd} 
\usepackage{xfrac}
\usepackage{pifont}
\usepackage{qcircuit}
\usepackage{enumerate}

\input{preamble}

\renewcommand\dots{\makebox[.7em][c]{.\hfil.\hfil.}}

\usetikzlibrary{intersections,fit,shapes.misc, decorations.markings,patterns}

\makeatletter
\def\namedlabel#1#2{\begingroup
    #2%
    \def\@currentlabel{#2}%
    \phantomsection\label{#1}\endgroup
}
\makeatother

\renewcommand\theoremautorefname{Th.}

\makeatletter
\patchcmd{\ALG@step}{\addtocounter{ALG@line}{1}}{\refstepcounter{ALG@line}}{}{}
\newcommand{\ALG@lineautorefname}{Line}
\makeatother

\renewcommand\phi{\varphi}

\newcommand\defaccr[2]{\newcommand#1{#2\xspace}}
\newcommand\defmath[2]{\newcommand#1{\ensuremath{#2}\xspace}}
\newcommand\concept[1]{\textit{#1}}

\defmath\leqpoly{\leq_\poly}

\defmath{\img}{\mathtt{image}}
\defmath{\apre}{\mathtt{\forall preimage}}

\defmath{\Post}{\mathit{Postfix}}

\let\set\undefined

\providecommand{\set}[1]{\ensuremath{\left\lbrace #1 \right\rbrace}}

\providecommand{\sizeof}[1]{\ensuremath{\left\vert{#1}\right\vert}}

\newcommand{\id}[1][]{\ensuremath{\mathbb{I}_{#1}}\xspace}
\defmath{\bool}{\ensuremath{\mathbb{B}}}
\defmath{\complex}{\ensuremath{\mathbb{C}}}
\defmath{\real}{\ensuremath{\mathbb{R}}}
\defmath{\integers}{\ensuremath{\mathbb{Z}}}
\defmath{\conditionalind}{\mathrel{\text{\scalebox{1.07}{$\perp\mkern-10mu\perp$}}}}
\defmath{\dx}{\partial x}
\defmath{\ddx}{\sfrac{\partial}{\partial x}}
\defmath{\half}{\textstyle{\frac{1}{2}}}

\newcommand{\underbraceset}[2]{\underset{#1}{\underbrace{#2}}}

\defmath\Exists{\mathit{Exists}}
\defmath\PlusExists{\mathit{PlusExists}}
\defmath\var{\mathit{var}}
\defmath\calciso{\mathsf{calciso}}

\newcommand{\defn}{\,\triangleq\,}

\tikzstyle{oval} = [state, ellipse, minimum size=4mm, inner sep=0.5mm, node distance=1cm]
\tikzset{every picture/.style={->,thick}}

\tikzstyle{leaf}=[draw, rectangle,minimum size=5mm, inner sep=3pt]
\tikzstyle{var}=[circle,draw=black!70,solid,thick,minimum size=6mm]
\tikzstyle{bdd}=[regular polygon, regular polygon sides=3, draw=black!70,solid,thick,inner sep=0.5mm]
\tikzstyle{n}=[->,loosely dashed,thick]
\tikzstyle{p}=[->,solid,thick]
\tikzstyle{b}=[->,densely dashdotted,ultra thick]

\defmath\before{\prec}
\defmath\beforeq{\preccurlyeq}

\newenvironment{smallmat}{\left[\begin{smallmatrix}}{\end{smallmatrix}\right]}

\newcommand\diag[1]{\begin{smallmat}1 & 0 \\  0 & #1\end{smallmat}\xspace}

\newcommand\diagg[2]{\begin{smallmat}#1 & 0 \\  0 & #2\end{smallmat}\xspace}

\defaccr{\bdd}{\textsf{BDD}}
\defaccr{\bdds}{\textsf{BDD}s}
\defmath{\qmdd}{\textsf{SLDD}_{\times}}
\defaccr{\add}{\textsf{ADD}}
\defaccr{\isoqmdd}{\textsf{LIMDD}}
\defaccr{\limdd}{\textsf{LIMDD}}
\defaccr{\zlimdd}{\ensuremath{\braket{Z}}-\textsf{LIMDD}}
\defaccr{\paulilimdd}{Pauli-\limdd}
\defaccr{\paulilimdds}{Pauli-\limdds}
\defaccr{\qmdds}{\textsf{QMDD}s}
\defaccr{\adds}{\textsf{ADD}s}
\defaccr{\isoqmdds}{\textsf{LIMDD}s}
\defaccr{\limdds}{\textsf{LIMDD}s}
\defaccr{\glimdd}{\ensuremath{G}-\limdd}
\defaccr{\glimdds}{\ensuremath{G}-\limdds}
\defaccr{\MPS}{\textsf{MPS}}

\newcommand\lbl{\textsf{label}\xspace}

\newcommand\Edge{\textsc{Edge}\xspace}
\newcommand\Node{\textsc{Node}\xspace}

\defmath\oh{\mathcal O}

\defmath\rootlim{B_{\textnormal{root}}}
\defmath\lowlim{B_{\textnormal{low}}}
\defmath\highlim{B_{\textnormal{high}}}
\defmath\gmax{g}
\defmath\kmax{\kappa^{\textnormal{final}}}

\defmath\cast{\mathbb C^\ast}

\defmath\plus{+}

\DeclareRobustCommand{\leafnode}[1][]{%
  \raisebox{-.8mm}{%
  \tikz{%
    \node[state,inner sep=0pt,minimum size=10pt,right= of x,leaf](v){\scriptsize $1$};%
  }%
  }%
}

\DeclareRobustCommand{\ledge}[3][]{%
  \raisebox{-.8mm}{%
  \tikz{%
    \node[inner sep=0pt] (x){$#1\,\,$};%
    \node[state,inner sep=0pt,minimum size=10pt,right= of x](v){\scriptsize $#3$};%
    \draw (x) to node[above,pos=.5]{\scriptsize $\,#2\,\,$} (v);%
  }%
  }%
}

\DeclareRobustCommand{\lnode}[5][]{%
  \raisebox{-1.5mm}{%
  \tikz{%
    \node[state,inner sep=0pt,minimum size=10pt] (v){\scriptsize $#1$};%
    \node[state,inner sep=0pt,minimum size=10pt,left= of v](v0){\scriptsize $#3$};%
    \draw[dotted] (v) to node[above,pos=.45]{\scriptsize $#2$} (v0);%
    \node[state,inner sep=0pt,minimum size=10pt,right= of v](v1){\scriptsize $#5$};%
    \draw (v) to node[above,pos=.45]{\scriptsize $#4$} (v1);%
  }%
  }%
}

\newcommand\low[1]{\ensuremath{#1[0]}}
\newcommand\high[1]{\ensuremath{#1[1]}}
\defmath\Low{\ensuremath{\textsf{low}}}
\defmath\High{\ensuremath{\textsf{high}}}
\defmath\LIM{\textsf{LIM}}

\tikzstyle{oval} = [state, ellipse, minimum size=4mm, inner sep=0.5mm, node distance=1cm]
\tikzset{every picture/.style={->,thick}}

\tikzstyle{leaf}=[draw, rectangle,minimum size=4.mm, inner sep=3pt]
\tikzstyle{var}=[circle,draw=black!70,solid,thick,minimum size=6mm]
\tikzstyle{bdd}=[regular polygon, regular polygon sides=3, draw=black!70,solid,thick,inner sep=0.5mm]
\tikzstyle{n}=[->,loosely dashed,thick]
\tikzstyle{p}=[->,solid,thick]
\tikzstyle{b}=[->,densely dashdotted,ultra thick]
\tikzset{every node/.style={initial text={}, inner sep=2pt, outer sep=0}}
\tikzstyle{e0}[0]=[dashed,thick,bend right=#1]
\tikzstyle{e1}[0]=[solid, bend left =#1]

\tikzstyle{lbl}=[draw,fill=white,inner sep=2pt, minimum size=0cm,line width=.5pt]

\defmath\yy{\begin{smallmat}
    0 & y^*\\
    y & 0\\
\end{smallmat}}

\defmath\ww{\begin{smallmat}
      0 & y   \\
      y^* & 0  \\
  \end{smallmat}
}

\def\paulilim{\textnormal{\sc PauliLIM}}
\tikzstyle{e0}[0]=[dotted,bend right=#1]
\tikzstyle{e1}[0]=[solid, bend left =#1]
\defmath\hv{{\hat v}}
\tikzset{every node/.style={initial text={}, inner sep=2pt, outer sep=0}}

\defmath\expsep{\succ\hspace{-1.5mm}\succ}

\newcommand{\allformulae}[2]{#1^{#2}}

\defmath{\sumstate}{\ket{\text{Sum}}}
\defmath{\atleastassuccinctas}{\preceq_s}%
\defmath{\strictlymoresuccinctthan}{\prec_s}
\defmath{\notmoresuccinctthan}{\npreceq_{s}}

\def\bigarrowhead{-{Latex[length=2mm,width=2mm]}}
\def\bigarrowheadb{{Latex[length=2mm,width=2mm]}-{Latex[length=2mm,width=2mm]}}

\newcommand\hide[1]{}

\newcommand\supp[1]{}

\def\crossedto{\hspace{3mm}\mathclap{\longrightarrow}{\hspace{-1.5mm}\times}\hspace{2mm}}

\def\strictsuccinctto{
    \setbox0\hbox{
            $\longrightarrow$
    }\copy0\llap{\raise\ht0\hbox{
    {
    $    \hspace{0mm}\mathclap{\longleftarrow}{\hspace{-1.5mm}\times}\hspace{0mm}$
    }
    }}
}

\newcommand{\nindex}{\ensuremath{\textsf{idx}}}

\newcommand{\Yes}{\ding{51}\xspace}
\newcommand{\Yar}{\ding{51}'\xspace}
\newcommand{\No}{\ding{54}\xspace}
\newcommand{\Cond}{\ensuremath{\circ}\xspace}
\defmath\samp{\textbf{Sample}}
\defmath\pro{\textbf{Measure}}
\defmath\gates{\textbf{Gates}}
\defmath\eq{\textbf{Equal}}
\defmath\res{\textbf{Res}}
\defmath\addi{\textbf{Addition}}
\defmath\inprod{\textbf{InnerProd}}
\defmath\fid{\textbf{Fidelity}}
\defmath\had{\textbf{Hadamard}}
\defmath\xyz{\textbf{X,Y,Z}}
\defmath\cx{\textbf{CX}}
\defmath\cz{\textbf{CZ}}
\defmath\swap{\textbf{Swap}}
\defmath\loc{\textbf{Local}}
\defmath\T{\textbf{T}}

\defmath\Rot{\mathit{Rot}}

\renewcommand\citeauthor[1]{\cite{#1}}

\definecolor{light-blue}{cmyk}{0.8,.15,.15,.15}
\newcommand\novelunderline[1]{\textcolor{light-blue}{\underline{#1}}}

\title{A Knowledge Compilation Map for Quantum Information}

\author{
	Lieuwe Vinkhuijzen,
	Tim Coopmans,
	Alfons Laarman
}
\affiliations{
	Leiden University, The Netherlands
}

\begin{document}

\crefname{equation}{Eq.}{Eq.}
\crefname{figure}{Fig.}{Fig.}
\crefname{section}{Sec.}{Sec.}
\crefname{algorithm}{Alg.}{Alg.}
\crefname{definition}{Def.}{Def.}
\crefname{theorem}{Th.}{Th.}
\crefname{lemma}{Lem.}{Lem.}
\crefname{appendix}{App.}{App.}
\crefname{tabular}{Table}{Tables}
\crefname{table}{Table}{Tables}
\crefname{corollary}{Cor.}{Corollaries}

\maketitle

\begin{abstract}
Quantum computing is finding promising applications in optimization, machine learning and physics, leading to the development of various models for representing quantum information.
Because these representations are often studied in different contexts (many-body physics, machine learning, formal verification, simulation), 
little is known about fundamental trade-offs between their succinctness and the runtime of operations to update them.
We therefore analytically investigate three widely-used quantum state representations: matrix product states (MPS), decision diagrams (DDs), and restricted Boltzmann machines (RBMs).
We map the relative succinctness of these data structures and provide the complexity for relevant query and manipulation operations.
Further, to chart the balance between succinctness and operation efficiency, we extend the concept of rapidity with support for the non-canonical data structures studied in this work, showing 
in particular that MPS is at least as rapid as some DDs.

By providing a knowledge compilation map for quantum state representations, this paper contributes to the understanding of the inherent time and space efficiency trade-offs in this area.

\end{abstract}

\section{Introduction}
\label{sec:introduction}
\label{sec:application-domains}

Quantum computers promise large computational advantages compared to classical computers in areas ranging from mathematics to chemistry.
To support this development, various classical methods have been proposed to analyze quantum systems.
A major bottleneck is that the state of a quantum system is described by an exponential-sized amplitude vector, and it is difficult to represent this object in general.
Therefore,  different disciplines, from physics to computer science,
have come up with data structures that enable quantum state representation and manipulation:
in order to simulate, verify and analyze new quantum algorithms, (algebraic) decision diagrams (DDs) have been proposed~\cite{miller2006qmdd,viamontes2003improving}.
To study physical quantum systems, physicists have developed tensor networks~(TNs; \cite{schollwock2011density}), matrix product states (MPS; \cite{orus2014practical}, sometimes also called tensor trains) ---a specialization of TNs---, and restricted Boltzmann machines (RBMs;~\cite{Dumoulin_Goodfellow_Courville_Bengio_2014}) with complex values, a specialization of Neural Quantum States \cite{carleo2017solving}.
Recently these methods have been growing towards each other, for instance, with the development of quantum circuit simulation for MPS~\cite{vidal2003efficient}, combinations of TNs and DDs~\cite{hong2020tensor}, MPS and decision diagrams~\cite{burgholzer2023tensor}, a comparison between TNs and probabilistic graphical models~\cite{glasser2019expressive},
and the extension of DDs with support for stabilizer states, yielding so-called \limdds~\cite{vinkhuijzen2021limdd}.

However, fundamental differences between these data structures  have not yet been studied in detail.
This choice between different structures introduces an important trade-off between \emph{size} and \emph{speed}, i.e., how much space the data structure uses, versus how fast certain operations, such as measurement, can be performed.
This trade-off plays a crucial role in other areas as well, and has already been illuminated for the domain of explainable AI~\cite{audemard2020tractable} and logic~\cite{darwiche2002knowledge,fargier2014knowledge}, but not yet for quantum information.

In this work, we analytically compare for the first time several data structures for representing and manipulating quantum states, motivated by the following three applications. 
First, \emph{simulating quantum circuits}, the building block of quantum computations, is a crucial tool for predicting the performance and scaling behavior of experimental devices with various error sources, thereby guiding hardware development~\cite{zulehner2018advanced,thanos2023fast}.
Next, \emph{variational methods} are the core of \emph{quantum machine learning}~\cite{Benedetti_2019,dunjko2017machine} and solve quantum physics questions such as finding the lowest energy of a system of quantum particles~\cite{foulkes2001quantum,carleo2017solving}.
Finally, \emph{verifying} if two quantum circuits are equivalent is crucial for checking if a (synthesized or optimized) quantum circuit satisfies its specification~\cite{ardeshir2014verification,gay2005communicating,burgholzer2020advanced}.

We focus on data structures developed for the above applications. In particular, we focus on: {algebraic decision diagrams} (\add), {semiring-labeled decision diagrams} (\qmdd, also called \textsc{QMDD}), {local invertible map decision diagrams} (\limdd), {matrix product states} (MPS) and restricted Boltzmann machines (RBM).
For those, we study \concept{succinctness}, \concept{tractability} and \concept{rapidity}:

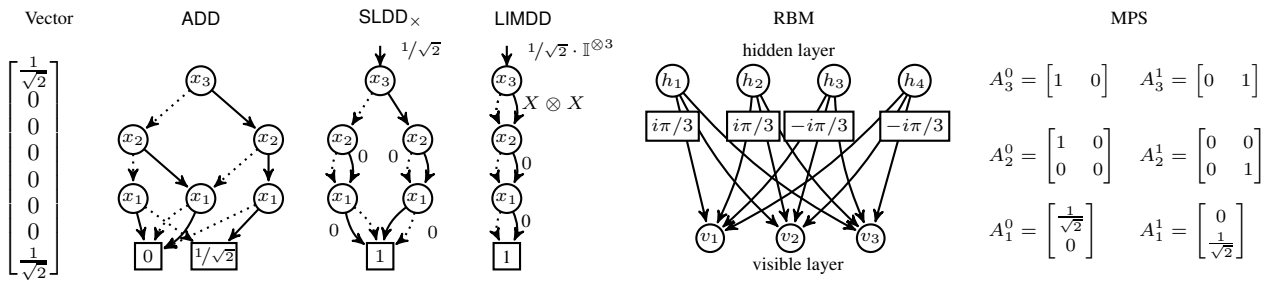
\begin{figure*}[b]
\begin{tikzpicture}[-{>[scale=0.3]},>=stealth',shorten >=1pt,auto,node distance=.4cm,
    thick, state/.style={circle,draw,minimum size=10pt,inner sep=.6pt},font=\scriptsize]

\node (vec) {
    \begin{minipage}{1.3cm}\footnotesize
$\def\arraystretch{1.}
    \begin{bmatrix*}[c]
    \frac1{\sqrt 2} \\ 0 \\ 0 \\ 0 \\ 0 \\ 0 \\ 0 \\ \frac1{\sqrt 2}
    \end{bmatrix*}
    $
    \end{minipage}
};

\node[state, right = 2cm of vec.north,anchor=north,yshift=-.2cm] (n1) {$x_3$};

\node[state](n2)[below = of n1, xshift=-.9cm]{$x_2$};
\node[state](n3)[below = of n1, xshift= .9cm]{$x_2$};

\node[state](n21)[below = of n2, ]{$x_1$};
\node[state](n22)[below = of n2, xshift= .9cm]{$x_1$};

\node[state](n31)[below = of n3]{$x_1$};

\node[draw, rectangle,minimum size=.5cm,below = of n21, xshift=.2cm,minimum size=10pt,inner sep=1pt] (l1) {$0$};
\node[draw, rectangle,minimum size=.5cm,right = of l1,minimum size=10pt,inner sep=1pt] (l2) {$\nicefrac1{\sqrt 2}$};

\path[]
(n1) edge[e1] node[right,pos=.7] {} (n3)
(n1) edge[e0] node[left,pos=.7] {} (n2)
(n2) edge[e0] node[left,pos=.7] {} (n21)
(n2) edge[e1] node[left,pos=.7] {} (n22)
(n3) edge[e0] node[left,pos=.7] {} (n22)
(n3) edge[e1] node[left,pos=.7] {} (n31)
(n21) edge[e0=  0]  node[pos=.7] {} (l2)
(n21) edge[e1=  0]  node[pos=.7] {} (l1)
(n22) edge[e0= 20]  node[pos=.7] {} (l1)
(n22) edge[e1= 20]  node[pos=.7] {} (l1)
(n31) edge[e0=  0]  node[pos=.7] {} (l1)
(n31) edge[e1=  0]  node[pos=.7] {} (l2)
;

\node[above=.8cm of n1,anchor=north]    (add) {{\add}};
\node[left =1.3cm of add]    (svc) {{Vector}};

\node[state, right=2cm of n1] (q1) {$x_3$} ;
\node[state, below=of q1,xshift=-0.5cm] (q2) {$x_2$} ;
\node[state, below=of q1,xshift=0.5cm] (q3) {$x_2$} ;
\node[state, below=of q2] (q4) {$x_1$} ;
\node[state, below=of q3] (q5) {$x_1$} ;
\node[draw, rectangle,minimum size=.5cm,below = of q4,xshift=0.5cm,minimum size=10pt,inner sep=1pt] (l1) {$1$};

\draw[e0] (q1) edge  node[] {} (q2);
\draw[e1] (q1) edge  node[] {} (q3);
\draw[e0=25] (q2) edge  node[] {} (q4);
\draw[e1=25] (q2) edge  node[,pos=.45] {$0$} (q4);
\draw[e0=25] (q3) edge  node[] {} (q5);
\draw[e1=25] (q3) edge  node[,xshift=-0.6cm,pos=.45] {$0$} (q5);
\draw[e0=-20] (q4) edge  node[] {} (l1);
\draw[e1=-20] (q4) edge  node[,left=.1cm,pos=.45] {$0$} (l1);
\draw[e0=-20] (q5) edge  node[] {} (l1);
\draw[e1=-20] (q5) edge  node[,xshift=0.4cm,pos=.45] {$0$} (l1);
\draw[<-] (q1) --++(90:.5cm) node[right=.2cm,pos=.7,] {$\nicefrac 1{\sqrt 2}$} node[right,pos=.8] {};

\node[right =1.7cm of add]    (sldd) {{\qmdd}};

    \node[state, right = 1.3cm of q1] (a1) {$x_3$};
    \node[state, below = of a1] (a3) {$x_2$};
    \node[state, below = of a3] (a4) {$x_1$};
    \node[draw,rectangle,minimum size=0.5cm, below= of a4,minimum size=10pt,inner sep=1pt] (w4) {1};

    \draw[<-] (a1) --++(90:.5cm) node[right=.2cm,pos=.7,] {$\nicefrac 1{\sqrt 2} \cdot \id^{\otimes3}$} node[left,pos=.8] {};
    \draw[e0=25] (a1) edge  node[] {} (a3);
    \draw[e1=25] (a1) edge  node[pos=.3,right] {$X \otimes X$} (a3);
    \draw[e0=25] (a1) edge  node[] {} (a3);
    \draw[e1=25] (a3) edge  node[pos=.3,right] {$0$} (a4);
    \draw[e0=25] (a3) edge  node[] {} (a4);
    \draw[e0=25] (a4) edge  node[] {} (w4);
    \draw[e1=25] (a4) edge  node[pos=.3,right] {0} (w4);

\node[right=0.8cm of sldd]    (limdd) {{\limdd}};

\node[state, right=1.8cm of a1,yshift=.cm] (h1) {$h_1$} ;
\node[state, right=.65cm of h1] (h2) {$h_2$} ;
\node[state, right=.65cm of h2] (h3) {$h_3$} ;
\node[state, right=.65cm of h3] (h4) {$h_4$} ;
\node[state, below=1.7cm of h1, xshift=0.5cm] (v1) {$v_1$} ;
\node[state, right=.7cm of v1] (v2) {$v_2$} ;
\node[state, right=.7cm of v2] (v3) {$v_3$} ;

\draw[e1] (h1) edge  node[] {} (v1);
\draw[e1=-10] (h1) edge  node[] {} (v2);
\draw[e1=-10] (h1) edge  node[] {} (v3);
\draw[e1=10] (h2) edge  node[] {} (v1);
\draw[e1] (h2) edge  node[] {} (v2);
\draw[e1=-10] (h2) edge  node[] {} (v3);
\draw[e1=10] (h3) edge  node[] {} (v1);
\draw[e1] (h3) edge  node[] {} (v2);
\draw[e1=-10] (h3) edge  node[] {} (v3);
\draw[e1=10] (h4) edge  node[] {} (v1);
\draw[e1=10] (h4) edge  node[] {} (v2);
\draw[e1] (h4) edge  node[] {} (v3);

\node[draw,rectangle, below=0.2cm of h1, fill=white, opacity=1] (1) {$i\pi/3$};
\node[draw,rectangle, below=0.2cm of h2, fill=white, opacity=1] (1) {$i\pi/3$};
\node[draw,rectangle, below=0.2cm of h3, xshift=-0.2cm, fill=white, opacity=1] (1) {$-i\pi/3$};
\node[draw,rectangle, below=0.2cm of h4, fill=white, opacity=1] (1) {$-i\pi/3$};

\node[above=0.cm of h3,xshift=-0.6cm]    (a31) {{hidden layer}};
\node[below=0.cm of v2,xshift=0.1cm]    (a31) {{visible layer}};

\node[right=2.8cm of limdd]    (rbm) {{RBM}};

\node[right=.7cm of h4] (a30) {$A_3^0=\begin{bmatrix}1 & 0\end{bmatrix}$};
\node[right=2cm of a30.west, anchor=west]   (a31) {$A_3^1=\begin{bmatrix}0 & 1\end{bmatrix}$};
\node[below=1.cm of a30.west, anchor=west] (a20) {$A_2^0=\begin{bmatrix}1 & 0 \\ 0 & 0\end{bmatrix}$};
\node[below=1.cm of a31.west, anchor=west] (a21) {$A_2^1=\begin{bmatrix}0 & 0 \\ 0 & 1\end{bmatrix}$};
\node[below=1.cm of a20.west, anchor=west] (a10) {$A_1^0=\begin{bmatrix}\frac1{\sqrt 2} \\ 0\end{bmatrix}$};
\node[below=1.cm of a21.west, anchor=west] (a11) {$A_1^1=\begin{bmatrix}0 \\ \frac1{\sqrt 2}\end{bmatrix}$};

\node[right=3.8cm of rbm]    (a31) {{MPS}};

\end{tikzpicture}
\caption{The $3$-qubit GHZ state $\nicefrac{1}{\sqrt 2}(\ket{000} + \ket{111})$, displayed using different data structures.
The unlabelled edges for \add, \qmdd, \limdd have resp. label 1, 1, $\id$.
In the RBM, the weights of edges incident to $h_1,h_2$ ($h_3, h_4$) are all $i\pi/3$ ($-i\pi/3$); the hidden node biases $(\beta_{h_1}, \beta_{h_2}, \beta_{h_3}, \beta_{h_4}) = i\pi \cdot (1/3, 2/3, -1/3, -2/3)$; the visible node biases $\alpha_{v_1}=\alpha_{v_2}=\alpha_{v_3}=0$.
}
\label{fig:ghz-examples}
\end{figure*}

\paragraph{Succinctness \& Tractability.}
In \cref{sec:succinctness-separations}, we find that the succinctness of MPS, RBM and \limdd are incomparable.
We also find that MPS is strictly more succinct than \qmdd, whereas previous research had suggested that these were incomparable \cite{burgholzer2023tensor}.
In \cref{sec:tractable-queries}, we give states which \qmdd and \limdd can compactly represent, but which take exponential time to apply Hadamard and swap gates.
Finally, we prove that computing fidelity and inner product in RBM and \limdd is intractable assuming the exponential time hypothesis.
An operation on a data structure is said to be \emph{tractable} if it runs in polynomial time in the size of the input \cite{darwiche2002knowledge}.

\paragraph{Rapidity.}
Considering tractability and succinctness separately can deceive; for instance, when representing a quantum state as an amplitude vector, most operations (e.g., applying a single-qubit gate) are polynomial time in the size of the size of this vector; thus, these exponential-time operations would be called tractable. Of course, this is only because this vector has exponential size in the number of qubits.
To mend this deficiency,  \citeauthor{lai2017new} introduced the notion of \concept{rapidity}, which reflects the idea that exponential operations may be preferable if a representation can be exponentially more succinct.

In \cref{sec:rapidity}, we generalize the definition of rapidity for non-canonical data structures, which allows us to study rapidity of MPS and RBM.
We also give a simple sufficient condition for data structures to be more rapid than others (for all operations).
We use it to settle several rapidity relations, showing surprisingly that MPS is strictly more rapid than \qmdd. %
Since we are unaware of a previous comparison, our knowledge was consistent with them being incomparable.

Proofs of all claims are provided in the appendix.

\section{Data Structures for Quantum States}
\label{sec:preliminaries}

\subsection{Quantum Information}
\label{sec:quantum-intro-in-main-text}

To understand our main results, it suffices to know the following
(we refer to \cref{sec:prelims2} and \citeauthor{nielsen2000quantum} for a more elaborate treatment of quantum information and computing).
A quantum state $\phi$ on $n$ quantum bits or qubits, written in Dirac notation as $\ket{\phi}$, can be represented as a vector of $2^n$ complex \emph{amplitudes} $a_k \in \mathbb{C}$  such that the vector has unit norm, i.e. $\sum_{k=1}^{2^n} |a_k|^2 = 1$, where $|z| = \sqrt{a^2 + b^2}$ is the modulus of a complex number $z = a + b\cdot i$ with $a, b\in \mathbb{R}$.
We also write $a_k = \langle k | \phi\rangle$.
Two quantum states $\ket{\phi}, \ket{\psi}$ are equivalent if $\ket{\phi} = \lambda \ket{\psi}$ for some complex $\lambda$.
The inner product between two quantum states $\ket{\phi}, \ket{\psi}$ is $\langle \phi | \psi \rangle = \sum_{k=1}^{2^n} \left(\langle k | \phi\rangle\right)^* \cdot \langle k | \psi\rangle$ where $z^* = a - bi$ is the complex conjugate of 
the complex number $z = a+b\cdot i$.
The fidelity $|\langle \phi | \psi \rangle|^2 \in [0, 1]$ is a measure of closeness, with fidelity equalling 1 if and only if $\ket{\phi}$ and $\ket{\psi}$ are equivalent.

A quantum circuit consists of gates and measurements.
An $n$-qubit \emph{gate} is a $2^n \times 2^n$ unitary matrix which updates the $n$-qubit state vector by matrix-vector multiplication.
A $k$-local gate is a gate which effectively only acts on a subregister of $k\leq n$   qubits.
Often-used gates go by names Hadamard ($H$), $T$, Swap, and the Pauli gates $X, Y, Z$. Together with the two-qubit gate controlled-$Z$, the $H$ and $T$ gate form a universal gate set (i.e. any quantum gate can be arbitrarily well approximated by circuits of these gates only).
Next, a (computational-basis) measurement is an operation which samples $n$ bits from (a probability distribution defined by) the quantum state,
\emph{while also updating the state.}

We now elaborate on  the application domains from \cref{sec:application-domains}.
\concept{Simulating a circuit} using the data structures in this work starts wlog by constructing a data structure for some (simple) initial state $\ket{\phi_0}$, followed by manipulating the data structure by applying the gates and measurements $A,B,\dots$ of the circuit one by one, i.e.: $\ket{\phi_1} =A\ket{\phi_0}$,  $\ket{\phi_2} =B\ket{\phi_1}$, etc.
The data structure supports strong (weak) simulation, if, for each measurement gate, it can produce (a sample of) the probability function (as a side result of the state update).
Next, the quantum circuit of \concept{variational methods} consist of a quantum circuit; then the output state $\ket{\phi}$ is used to compute $\langle \phi | O | \phi\rangle$ for some linear operator $O$ (an \concept{observable}).
This computation reduces to $\langle \phi | \psi \rangle$ with $\ket{\psi} := O\ket{\phi}$ (i.e., computing simulation and fidelity).
Last, \concept{circuit verification} relies on checking approximate or exact equivalence of quantum states.
This extends to unitary matrices (gates), which all data structures from this paper can represent 
also, but which we will not treat for simplicity.

\subsection{Data Structures}
\label{sec:data-structures-definitions}
We now define the data structures for representing quantum states considered in this work, with examples in \cref{fig:ghz-examples}.

\begin{definition}
[Inspired by \protect\citeauthor{fargier2014knowledge}]
A \concept{quantum-state representation (language)} %
 is a tuple $(D, n, \ket{.}, |.|)$ where $D$ is a set of data structures. Given $\alpha\in D$, $\ket{\alpha}$ is the (possibly unnormalized) quantum state it represents (i.e., the interpretation of $\alpha$), $|\alpha|$ is the size of the data structure, and $n(\alpha)$, or $n$ in short, is the number of qubits~of~$\ket{\alpha}$.
Finally, each quantum state should be expressible in the language.
\end{definition}

We will often refer to a representation as a data structure.
We define $\allformulae{D}{\phi} \defn \{\alpha \in D \mid \ket{\alpha} = \ket{\phi}\}$, i.e.,  the set of all data structures in language $D$ representing state~$\ket{\phi}$.
In line with quantum information\supp{ (see \cref{sec:quantum-intro})}, we say that data structures $\alpha, \beta$ are \concept{equivalent} if $\ket{\alpha} = \lambda \ket{\beta}$ for some $\lambda \in \mathbb{C}$.%

\paragraph{Vector.} 
A state vector $\alpha$ is a length-$2^n$ array of complex entries $a_k \in \mathbb{C}^{2^n}$ satisfying $\sum_{k=1}^{2^n} |a_k|^2 = 1$,  and interpretation $\ket \alpha = \sum_j \alpha_j\ket j$.
Despite its size, the vector representation is used in many simulators~\cite{jones2019quest}.
We mainly include the vector representation to show that considering operation tractability alone can be deceiving.

\paragraph{Matrix Product States (MPS).}
An MPS $M$ is a series of $2n$ matrices $A_k^x \in \mathbb{C}^{D_k \times D_{k-1}}$ where $x\in \{0, 1\}$, $k\in [n]$ and $D_k$ is the row dimension of the $k$-th matrix with $D_0 = D_{n} = 1$.
The interpretation $\ket{M}$ is determined as $\braket{\vec y | M} = A_n^{x_n} \cdots A_2^{x_2} A_1^{x_1}$ for $\vec y \in \{0, 1\}^n$.
The size of $M$ is the total number of matrix elements, i.e., $|M| = 2\cdot \sum_{k=1}^{n} D_{k} \cdot D_{k-1}$.
We will speak of $\max_{j\in [0 \dots n]} D_j$ as `the' bond dimension.

\paragraph{Restricted Boltzmann Machine (RBM).}
An $n$-qubit RBM is a tuple $\mathcal M=(\vec \alpha,\vec\beta, W,m)$,
where $\vec \alpha\in \mathbb C^n, \vec\beta\in\mathbb C^m$ for $m \in \mathbb{N}_{+}$ are \concept{bias vectors} and $W \in \mathbb{C}^{n \times m}$ is a \concept{weight matrix}.
An RBM $\mathcal M$ represents the state $\ket{\mathcal M}$ as follows.
\begin{align}
	\label{eq:rbm-coefficient}
	\braket{\vec x|\mathcal M}
	= & e^{\vec x^T\cdot \vec\alpha}\cdot \prod_{j=1}^m (1 + e^{\beta_j + {\vec x}^T\cdot {\vec{W}}_j})
\end{align}
where $\vec W_j$ is the $j$-th column of $W$, $\beta_j$ is the $j$-th entry of $\beta$ and where we write $\vec x^T\cdot W_j$ to denote the inner product of the row vector $\vec x^T$ and the column vector $\vec{W}_j$ \cite{chen2018equivalence}.
The size of $\mathcal{M}$ is $|\mathcal{M}| = n + m + n \cdot m$.
We say this RBM has $n$ \concept{visible nodes} and $m$ \concept{hidden nodes}.
A weight $W_{v,j}$ is an edge from the $v$-th visible node to the $j$-th hidden node.
The $j$-th hidden node is said to contribute the multiplicative term $(1 + e^{\beta_j + \vec x^T\cdot {\vec{W}}_j})$.

\newcommand{\edgelabels}{\mathcal{E}}
\newcommand{\leaflabels}{\mathcal{L}}
\paragraph{Quantum Decision Diagrams (QDD).} %

\begin{table}[t!]
\caption{
Various decision diagrams treated by the literature.
The column \emph{Merging strategy} lists the conditions under which two nodes $v,w$, representing subfunctions $f,g\colon \{0,1\}^k\to\mathbb C$ are merged.
Here $p,a\in \mathbb C$ are complex constants, $P_i$ are Pauli gates and
$f + a$ means the function $f(\vec x) + a$ for all $\vec x$.
We list SLDD$_+$ and AADD for completeness, but do not consider them.
}
\label{tab:dds}
\centering\scriptsize
\def\arraystretch{.9}
\begin{tabular}{|p{5.2cm}|c|}
\hline
\textbf{(Quantum) decision diagrams (and variants)}				& \textbf{Node merging strategy}   \\\hline
Decision Tree	& (no merging)  \\
\setstretch{.6} \add~\cite{bahar1997algebric}, MTBDD~\cite{clarke1993spectral}, QuiDD~\cite{viamontes2003improving}			& $f = g$   \\
\setstretch{.6} SLDD$_\times$~\cite{wilson2005decision,fargier2013semiring}, QMDD~\cite{miller2006qmdd} %
	& $f = p\cdot g$   \\
\setstretch{.6} SLDD$_+$~\cite{fargier2013semiring}, EVBDD~\cite{lai1994evbdd}	& $f = g + a$   \\
\setstretch{.6} AADD~\cite{sanner2005AffineADDs}, FEVBDD~\cite{tafertshofer1997factored} 		& $f = p\cdot g + a$  \\
LIMDD~\cite{vinkhuijzen2021limdd}			& $f = p  P_1\otimes \dots \otimes P_n \cdot g$    \\
\hline
\end{tabular}
\end{table}

We define a Quantum Decision Diagram to represent a quantum state, based on the Valued Decision Diagram~\cite{fargier2014knowledge} as instantiated on a domain of binary variables (qu\textbf{bit}s) and a co-domain of complex values (amplitudes).
A QDD $\alpha$ is a finite, rooted, directed acyclic graph $(V, E)$, where each node $v$ is labeled with an index $\nindex(v) \in [n]$ and leaves have index $0$.
In addition, a `root edge' $e_R$ (without a source node) points to the root node.
Each node has two outgoing edges, one labeled $0$ (the low edge) and one labeled $1$ (the high edge).
In addition, edge $e = vw$ pointing to a node $w$ with index $k$ has a label $label(e) \in \edgelabels_k$ for some edge label set $\edgelabels_k$; in this paper, $\edgelabels_k$ is a group (for \qmdd and \limdd below with $0$ added).
Also, each leaf node $v$ has a label $\lbl(v) \in \leaflabels$.
The size of a QDD is $|\alpha| = |V| + |E| + \sum_{v \in V} |\lbl(v)| + \sum_{e\in E \cup\{e_R\}} |\lbl(e)|$.
For simplicity, we require that no nodes are skipped, i.e. $\forall vw\in E \colon  id[v] = id[w] + 1$.\footnote{Asymptotic analysis is not affected by disallowing node skipping as it yields linear-size reductions at best~\cite{knuth4}.}
The semantics are:
\begin{itemize}
\item A leaf node $v$ represents the value $\lbl(v)$.
\item A non-leaf node $\lnode[v]{e_{\mathit{low}}}{}{e_{\mathit{high}}}{}$ represents\\ $\ket{v} = \ket{0} \otimes \ket{e_{\mathit{low}}} + \ket{1} \otimes \ket{e_{\mathit{high}}}$.
\item An edge $\ledge{e}{v}$ represents $\ket{e} = \lbl(e) \cdot \ket{v}$.
\end{itemize}

Consequently, any node $v$ at \concept{level} $k$, i.e., with $\nindex(v) = k$, is $k$ edges away from a leaf (along all paths) and therefore represents a $k$-qubit quantum state (or a complex vector).

In this paper, we consider the following types of QDDs.
We emphasize that an \add can be seen as a special case of \qmdd, which is a special case of \limdd.
\begin{description}
\item[\add:] ~~~~~$\forall k\colon  \edgelabels_k = \set{1},~ \leaflabels=\mathbb{C}$ \cite{bahar1997algebric}.
\item[\qmdd:] $\forall k\colon  \edgelabels_k = \mathbb{C},~ \leaflabels=\{1\}$ \cite{wilson2005decision}.
\item[\limdd:] ~$\forall k\colon \edgelabels_k = \paulilim_k \cup \{0\},~ \leaflabels=\{1\}$,~where 
	\hbox{$\paulilim_k \defn \{\lambda P \mid \lambda \in \mathbb{C} \setminus\{0\}, P \in \set{\id,X,Y,Z}^{\otimes k}\}$}, i.e., the group generated by $k$-fold tensor products of the single-qubit Pauli matrices \cite{vinkhuijzen2021limdd}:
\[
\id \hspace{-.5ex}\defn\hspace{-.5ex} \begin{pmatrix} 1 & 0\\ 0 & 1 \end{pmatrix}\hspace{-.5ex},
X \hspace{-.5ex}\defn\hspace{-.5ex} \begin{pmatrix} 0 & 1\\ 1 & 0 \end{pmatrix}\hspace{-.5ex},
Y \hspace{-.5ex}\defn\hspace{-.5ex} \begin{pmatrix} 0 & \hspace{-1.3ex}-i\\ i & 0 \end{pmatrix}\hspace{-.5ex},
Z \hspace{-.5ex}\defn\hspace{-.5ex} \begin{pmatrix} 1 & 0\\ 0 & \hspace{-1.3ex}-1 \end{pmatrix}
\]
\end{description}

Two isomorphic QDD nodes (\cref{def:iso}) $v, w$ with $\ket w = \ell \ket v$ can be \concept{merged} by removing $w$ and rerouting all edges $uw \in E$ incoming to $w$ to $v$, updating their edge labels
accordingly (i.e., $\lbl(uv) := \lbl(uw) \cdot \ell$).
\cref{tab:dds} summarizes merging strategy for the above QDDs and related versions.
If all isomorphic nodes are merged, we call a QDD \concept{reduced}.
We may assume a QDD is reduced (and even canonical), since it can be reduced in polynomial time and manipulation algorithms keep the QDD reduced using a \textsc{MakeNode} operation~\cite{bryant86,miller2006qmdd,fargier2013semiring,vinkhuijzen2021limdd}.

\begin{definition}[Isomorphic nodes]\label{def:iso}
QDD node $v$ is isomorphic to node $w$, if there exists an edge label $\ell \in \edgelabels_{\nindex(v)}, \ell \neq 0$ such that $\ell \cdot \ket{v}= \ket{w}$.
Since $\edgelabels_{\nindex(v)}$ is a group, $\ell \in \edgelabels_{\nindex(v)}$ implies $\ell^{-1} \in \edgelabels_{\nindex(v)}$ so $\ell^{-1}\ket{w}=\ket{v}$.
Hence isomorphism is an equivalence relation.
\end{definition}

It is well known that the variable order greatly influences QDD sizes~\cite{bollig1996improving,darwiche2011sdd,bova2016sdds}. Our results here assume any variable order: For instance, when we say that some structure is strictly more succinct than a QDD, it means that there is no variable order for which the QDD is more succinct than the other structure (for representing a certain worst-case state).

\section{Succinctness~of~Quantum~Representations}
\label{sec:succinctness-separations}

\let\oldsupp\supp
\renewcommand\supp[1]{}
\begin{figure}[b!]
\centering
\vspace{-1em}
\scalebox{.7}{
\begin{tikzpicture}[node distance=.7cm,minimum height=.5cm]
		\node[draw] (mps) 			   {MPS};
		\node[draw, right =2cm of mps, yshift=-1.5cm] (limdd) {\limdd};
		\node[draw, right =2cm of limdd, yshift=1.5cm] (rbm) {RBM};
		\node[draw, below = .5cm of limdd] (qmdd) {\qmdd};
		\node[draw, below = .5cm of qmdd] (add) {ADD};

		\draw[\bigarrowhead, bend left=20] (limdd.north west) to node[midway] {$\boldsymbol{\bigtimes}$}
											node[pos=0.3,left=1cm,below=0cm] {\textcolor{black}{\supp{\cref{thm:succ-limdd-vs-mps}}}}
											(mps.south east);
		\draw[\bigarrowhead, bend left=20, color=light-blue] (mps.south east) to node[midway] {$\boldsymbol{\bigtimes}$}
											node[pos=1,right=0cm,above=0.3cm] {\textcolor{black}{\supp{\cref{thm:succ-mps-vs-limdd}}}}
											(limdd.north west);

		\draw[\bigarrowhead, bend left=10, color=light-blue] (qmdd.west) to node[pos=.3,below left=-.5cm]
									{\textcolor{black}{\supp{\cref{thm:mps-exp-more-succinct-than-qmdd}}}} (mps);

		\draw[\bigarrowhead, bend left=20, color=light-blue] (add.west) to
													 (mps.south west);

		\draw[\bigarrowhead, bend left=10] (rbm.west) to node[midway] {$\boldsymbol{\bigtimes}$} 
													node[midway,above=0.125cm] {\supp{\cref{thm:succ-rbm-vs-mps}}}
											(mps.east);
		\draw[\bigarrowhead, bend left=10, color=light-blue] (mps.north east) to node[midway] {$\boldsymbol{\bigtimes}$} 
													node[midway,above=0.125cm] {\textcolor{black}{\supp{\cref{thm:add-rbm}}}}
											(rbm.north west);

		\draw[\bigarrowheadb, bend right=0, color=light-blue] (rbm.south west) to node[midway] {$\boldsymbol{\bigtimes}$}
											node[pos=.7,above left=-1mm] {\textcolor{black}{\supp{\cref{thm:succ-rbm-vs-limdd}}}}
											(limdd.north east);

		\draw[\bigarrowheadb, bend left=10, color=light-blue] (rbm) to node[midway] {$\boldsymbol{\bigtimes}$}
										   (qmdd.east);

		\draw[\bigarrowheadb, bend left=20, color=light-blue] (rbm.south east) to node[midway] {$\boldsymbol{\bigtimes}$}
											node[pos=.3,below right=-1mm] {\textcolor{black}{\supp{\cref{thm:add-rbm}}}} 
											 (add.east);

		\draw[\bigarrowhead, bend right=0] (qmdd) to  node[pos=.5,right=.01cm] 
									{\supp{\cref{thm:succ-limdd-more-qmdd}}}	 (limdd);
		\draw[\bigarrowhead, bend right=0] (add) to node[pos=.5,right=.01cm]
									{\supp{\cref{thm:succ-qmdd-more-add}}} 	(qmdd);

	\end{tikzpicture}
}
\caption{Succinctness relations between various classical data structures for representing quantum states.
Solid arrows $A\to B$ denote $B\strictlymoresuccinctthan A$, i.e., $B$ is strictly more succinct than~$A$.
Crossed arrows $A\crossedto B$ denote a separation $B \notmoresuccinctthan A$; a bidirectional crossed arrow implies incomparability.
Blue arrows indicate novel relations that we identified.
}
\label{fig:succinct}
\end{figure}
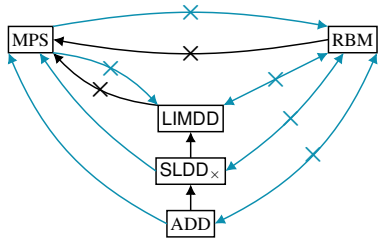
\renewcommand\supp\oldsupp

In this section, we compare the succinctness of the data structures introduced in \cref{sec:preliminaries}.
A language $L_1$ is as succinct as $L_2$ if $L_2$ uses at least as much space to represent any quantum state as $L_1$, up to polynomial factors.
Thus, a more succinct diagram is more expressive when we constrain ourselves to polynomial (or asymptotically less) space.
We define a language $L_1$ to be \concept{at least as succinct} as $L_2$ (written $L_1 \atleastassuccinctas L_2$), if there exists a polynomial $p$ such that for all $\beta\in L_2$, there exists $\alpha \in L_1$ such that $\ket{\alpha} = \ket{\beta}$ and $|\alpha| \leq p(|\beta|)$.
We say that $L_1$ is more succinct ($L_1 \strictlymoresuccinctthan L_2$) if $L_1$ is at least as succinct as $L_2$ but not vice versa.

The results of this section are summarized by \cref{thm:succinctness} (\cref{fig:succinct})\supp{ and proved in \cref{sec:proofs-of-sec-succinctness}}.
We now highlight our novel results.
First, we show that MPS \strictlymoresuccinctthan \qmdd, by (i) describing a simple and efficient way to write a \qmdd as an MPS, and (ii) finding a family of states, which we call \sumstate, which have small MPS but which require exponentially large \qmdd. It follows that MPS \strictlymoresuccinctthan \add.
Second, we strengthen (ii) to show that the same state also requires exponentially large \limdds, thus we establish \limdd \notmoresuccinctthan MPS. The reverse here also holds because MPS cannot efficiently represent certain stabilizer states, while \limdd is small for all stabilizer states, as shown in~\cite{vinkhuijzen2021limdd}.
Lastly, we show that RBMs can efficiently represent \sumstate.
Since it is well-known that RBM explodes for parity functions, which are small for QDD~\cite{martens2013representational}, we establish incomparability of RBM with all three QDDs.

\begin{restatable}{theorem}{succinctness}
\label{thm:succinctness}
The succinctness results in \cref{fig:succinct} hold.
\end{restatable}

\section{Tractability of Quantum Operations}
\label{sec:tractable-queries}

In this section, we investigate for each data structure (DS) the tractability of the main relevant queries and manipulation operations for the different applications discussed \cref{sec:application-domains}.

By a \concept{manipulation operation}, we mean a map $D^c \rightarrow D$ and by a \concept{query operation} a map $D^c  \rightarrow \mathbb{C}$, where $D$ is a class of data structures and $c \in \mathbb N_{\ge 1}$ the number of operands.
We say that a class of data structures $D$ \emph{supports} a (query or manipulation) operation $OP(D)$, if there exists an algorithm implementing $OP$ whose runtime is polynomial in the size of the operands, i.e.,  $\sizeof{\phi_1} + \dots  + \sizeof{\phi_c}$.

\newcolumntype{R}[2]{%
    >{\adjustbox{angle=#1,lap=\width-(#2)}\bgroup}%
    l%
    <{\egroup}%
}
\newcommand\rot[1]{\multicolumn{1}{R{90}{0em}|}{#1}}

The operations whose tractability we investigate are:
\begin{itemize}
	\item \samp: Given a  DS  representing $\ket\phi$, sample the measurement outcomes $\vec x\in\{0,1\}^n$ from measuring all qubits of $\ket\phi$ in the computational basis.
	\item \pro: Given a DS representing $\ket\phi$ and $\vec x\in\{0,1\}^n$, compute the probability to get $\vec x$ when measuring $\ket\phi$.
	\item \textbf{Gates} (\had, Pauli \xyz, Controlled-Z \cz, \swap, \T): Given a  DS  representing $\ket\phi$ and a one- or two-qubit gate $U$, construct a DS representing $U\ket\phi$. This gate set is universal \cite{bravyi2005universal} and is used by many  algorithms \cite{vandaele2023optimal}.
	\item \loc (Local gates) as general case of Gates: Given a DS representing $\ket\phi$, a constant $k\in\mathbb{N}_{\geq 1}$ and a $k$-local gate $U$, construct a DS representing the state $U\cdot\ket\phi$.
	\item \addi: Given DSs for states $\ket\phi,\ket\psi$, construct a DS representing the state $\ket\phi + \ket\psi$.
	\item \inprod (inner product) and \textbf{Fidelity}: Given DSs for states $\ket\psi,\ket\psi$, compute $\braket{\phi|\psi}$ and $|\braket{\phi|\psi}|^2$.
	\item \eq (Equality): Decide whether the states represented by two given data structures are equivalent.
\end{itemize}

\begin{table}[t!]
\caption{Tractability of queries and manipulations on the data structures analyzed in this paper (single application of the operation). 
A \Yes means the data structure supports the operation in polytime, a \Yar means supported in randomized polytime, and \No means the data structure does not support the operation in polytime.
A \Cond means the operation is not supported in polytime unless $P=NP$.
? means unknown.
The table only considers deterministic algorithms (for some ? a probabilistic algorithm exists, e.g., for \inprod on RBM).
Novel results are blue and underlined.
}
\setlength{\tabcolsep}{4pt}
\def\arraystretch{1.1}
\footnotesize
\centering
\begin{tabular}{|l@{\hspace{10pt}}|| *{5}{c|}| *{20}{c|}}
\hline
 & \multicolumn{5}{c||}{Queries} & \multicolumn{7}{c|}{Manipulation operations} \\
	& \rot{\samp} & \rot{\pro} & \rot{\eq}  & \rot{\inprod} & \multicolumn{1}{R{90}{0em}||}{\fid}
	& \rot{\addi} & \rot\had & \rot{\xyz} & \rot\cz & \rot{\swap} & \rot{\loc} & \rot{\T-gate} \\
\hline
Vector& \Yar & \Yes & \Yes & \Yes & \Yes & \Yes & \Yes & \Yes & \Yes & \Yes & \Yes & \Yes \\
\hline
ADD \supp{~~~~~~~~~~\hfill \cref{app:trac:add}}  	& \Yar	& \Yes	& \Yes	& \Yes & \Yes
		& \Yes	& \Yes 	& \Yes	& \Yes	& \Yes	& \Yes	& \Yes \\
\hline
\qmdd \supp{\hfill \cref{app:trac:qmdd}}
		& \Yar	& \Yes	& \Yes	& \novelunderline{\Yes} & \novelunderline{\Yes}
		& \No	& \novelunderline{\No} 	& \novelunderline{\Yes}	& \novelunderline{\Yes}	& \novelunderline{\No}	& \novelunderline{\No}	& \novelunderline{\Yes} \\
\hline 
	\limdd \supp{\hfill \cref{app:trac:limdd}}
	 	& \Yar	& \Yes	& \Yes	& \novelunderline{\Cond} & \novelunderline{\Cond}
		& \novelunderline{\No}	& \novelunderline{\No}	& \Yes	& \Yes	& \novelunderline{\No}	& \novelunderline{\No}	& \novelunderline{\Yes}  \\
\hline 
MPS   \supp{\hfill \cref{app:trac:mps}}
	  & \Yar & \Yes & \Yes & \Yes & \Yes & \Yes 
	  & \Yes & \Yes & \Yes & \Yes & \Yes & \Yes  \\
\hline 
	RBM   \supp{\hfill \cref{app:trac:rbm}}
	 & \Yar    & ? & ? & \novelunderline{\Cond} & \novelunderline{\Cond} & ? & ? & \Yes & \Yes & \Yes & ? & \Yes \\
\hline 
\end{tabular}
\label{tab:tractability}
\end{table}

We are motivated to study these operations by their applicability in the three application domains from \cref{sec:application-domains}.
First, classical simulation of quantum circuits includes \gates, as well as \pro (strong simulation) and \samp (weak simulation).
Although \addi is not, technically speaking, a quantum operation, we include it because addition of quantum states can happen due to quantum operations.
For example, if we apply first a Hadamard gate, and then a measurement, to the state $\ket0\ket\phi + \ket1\ket\psi$, we may obtain the state $\ket 0\otimes (\ket\phi+\ket\psi)$.
For a data structure, this resulting state is at least as difficult to represent as the sum $\ket\phi + \ket\psi$.
Therefore, it is instructive to check which data structures support \addi: it allows us to explain why certain gates are not tractable for a data structure (such as the swap and general local gates).
In particular, all decision diagram implementations support an explicit addition subroutine \cite{miller2006qmdd,Clarke2001}.
In addition, \inprod and \fid (and \loc) are required in variational methods and quantum circuit verification involves \textbf{Eq}.

For the operations listed above and the languages from \cref{sec:preliminaries}, we present an overview of existing and novel tractability results in \cref{thm:tractability} (\cref{tab:tractability})\supp{, with proofs for all entries in \cref{sec:proofs-of-sec-tractability}}.
The novel results in this work are the hardness results (denoted \Cond)  for \inprod and \fid  on  \limdd and RBM (\cref{thm:fidelity} and for \inprod we can reduce from \fid)
and unsupported manipulations (denoted \No) on \qmdd and \limdd.
More generally, our proof shows that computing \fid is hard for any data structure which succinctly represents both all graph states and the Dicke state, such as RBM and \limdd.
A fidelity algorithm for \qmdd was mentioned by \cite{burgholzer2021random}, but, to the best of our knowledge, had not previously been described or analyzed.
\begin{restatable}{theorem}{thmtractability}
\label{thm:tractability}
The tractability results in \cref{tab:tractability} hold.
\end{restatable}
\begin{restatable}{theorem}{thmfidelity}
\label{thm:fidelity}
	Assuming the exponential time hypothesis, the fidelity of two states represented as \limdds or RBMs cannot be computed in  polynomial time.
	The proof uses a reduction from the \#\textsc{EVEN SUBGRAPHS} problem \cite{jerrum2017parameterised}.
\end{restatable}%

\section{Rapidity of Data Structures}
\label{sec:rapidity}

\renewcommand\time{\ensuremath{\mathit{time}}}

The tractability criterion, studied in \cref{sec:tractable-queries}, sometimes gives a skewed picture of efficiency.
For example, looking naively at \cref{tab:tractability}, it seems that \add is faster than \qmdd when applying a Hadamard gate.
Yet there is no state for which applying the Hadamard gate on its \qmdd representation is slower than on its \add representation.
So succinctness actually consistently mitigates the worst-case runtime behavior. 
To remedy this shortcoming, \citeauthor{lai2017new} introduced the notion of \concept{rapidity} for canonical data structures.

In \cref{def:rapidity}, we generalize rapidity to support non-canonical data structures, such as MPS, RBM, d-DNNF~\cite{darwiche2001decomposable} and CCDD~\cite{lai2022ccdd}.
To achieve this, \cref{def:rapidity} requires that for a fixed input $\phi$ to $ALG_1$, among all equivalent inputs to $ALG_2$, there is one on which  $ALG_2$ is at least as fast as $ALG_1$.
It may seem reasonable to require, instead, that $ALG_2$ is at least as fast as $ALG_1$ on \emph{all} equivalent inputs; however, in general, there may be no upper bound on the size of the (infinitely many) such inputs; thus, such a requirement would always be vacuously false.

In the following, 
we write $\time(A, x)$ for the runtime of algorithm $A$ on input~$ x$.

\begin{restatable}[Rapidity for non-canonical data structures]{definition}{defrapidity}
	\label{def:rapidity}
	Let $D_1,D_2$ be two data structures and consider some $c$-ary operation $OP$ on these data structures.
    In the below, $ALG_1$ ($ALG_2$) is an algorithm implementing $OP$ for $D_1$ ($D_2$).
\begin{enumerate}[(a)]

    \item We say that $ALG_1$ is \emph{at most as rapid as} $ALG_2$ iff there exists a polynomial $p$ such that
    for each input $\phi=(\phi_1,\ldots,\phi_c)$ there exists an equivalent input $\psi=(\psi_1,\ldots,\psi_c)$, i.e., with $\ket{\phi_j}=\ket{\psi_j}$ for $j=1\ldots c$, for which 
			$\time(ALG_2, \psi) \leq p\left(\time(ALG_1, \phi)\right)$.
			We say that $ALG_2$ is \emph{at least as rapid as} $ALG_1$.
	\label{item:rapidity-algorithm}
	\item \label{item:rapidity-operations} We say that $OP(D_1)$ is \emph{at most as rapid as} $OP(D_2)$ if for each algorithm $ALG_1$ performing $OP(D_1)$, there is an algorithm $ALG_2$ performing $OP(D_2)$ such that $ALG_1$ is at most as rapid as $ALG_2$.
\end{enumerate}
\end{restatable}

We remark that, when applied to canonical data structures, \cref{def:rapidity} reduces to the definition by Lai et al. \supp{\cref{thm:equivalence-of-rapidity-definitions}}, except that Lai et al. allow the input to be fully read by the algorithm: $\time(ALG_1,x_1)\leq \poly(\time(ALG_2,x_2) \underline{+ |x_2|})$ (difference underlined). We omit this to achieve transitivity.
Rapidity indeed has the desirable property that it is a preorder, i.e., it is reflexive and transitive, as \cref{thm:rapidity-is-preorder} shows. \supp{\cref{sec:proofs-of-sec-rapidity} provides a proof.}
\begin{restatable}{theorem}{rapidityispreorder}
	\label{thm:rapidity-is-preorder}
	Rapidity is a preorder over data structures.
\end{restatable}

\subsection{A Sufficient Condition for Rapidity}
\label{sec:sufficient}

We now introduce a simple sufficient condition for rapidity in \cref{thm:sufficient-condition-rapidity}, allowing researchers to easily establish that one data structure is more rapid than another 
 \emph{for many relevant operations simultaneously}.
Previously, such proofs were done for each operation individually \cite{lai2017new}.
We use this sufficient condition to establish rapidity relations between many of the data structures studied in this work.
By a \emph{transformation} $f$ from data structure $D_1$ to $D_2$ we mean a map such that $\ket{f(x_1)} = \ket{x_1}$ for all $x_1\in D_1$. We also need the notions of a weakly minimizing transformation (\cref{def:minimizing}) and a runtime monotonic algorithm (\cref{def:runtime-monotonic}).

\begin{definition}[Weakly minimizing transformation]
	\label{def:minimizing}
	Let $D_1,D_2$ be data structures.
	A transformation $f: D_1 \rightarrow D_2$
	is \emph{weakly minimizing} if $f$ always outputs an instance which is polynomially close to minimum-size, i.e., 
	there exists a polynomial $p$ such that for all $x_1\in D_1, x_2\in D_2$ with $\ket{x_1} = \ket{x_2}$, we have $\sizeof{f(x_1)} \leq p(\sizeof {x_2})$.
\end{definition}

\begin{definition}[Runtime monotonic algorithm]
	\label{def:runtime-monotonic}
	An algortihm $ALG$ implementing some operation on data structure $D$ is \emph{runtime monotonic} if for each polynomial~$s$ there is a polynomial $t$ such that for each state $\ket\phi$ and each $x,y\in D^{\phi}$, if  $\sizeof x \leq s(\sizeof y)$, then $\time(ALG,x) \leq t(\time(ALG,y))$.
\end{definition}

\newcommand{\rapidityproofoutputlabel}[2]{}

\begin{restatable}[b!]{figure}{sufficientconditionprooffigure}
 \centering
 \scalebox{.8}{
 \begin{tikzpicture}[ele/.style={fill=black,circle,minimum width=.8pt,inner sep=1pt},every fit/.style={ellipse,draw,inner sep=-2pt}]

  \node[ele,label=left:$x_1$] (x1) at (0,2.5) {};    
	 \node[ele,label=left:\rapidityproofoutputlabel{1}{}] (a1) at (0,1) {};

  \node[ele,,label=left:$x_2$] (x2) at (6.5,2.5) {};
  \node[ele,,label=right:$f(x_1)$] (fx1) at (4,2.5) {};
	 \node[ele,,label=left:\rapidityproofoutputlabel{2}{}] (a2) at (6.5,1) {};
	 \node[ele,,label=right:\rapidityproofoutputlabel{2}{'}] (a2prime) at (4,1) {};

  \node[draw,gray!60,fit= (x1) (a1) ,minimum width=2cm,minimum height=3cm, label={above:$D_1$}] {};] {} ;
  \node[draw,gray!60,fit= (x2) (fx1) (a2prime) (a2),minimum width=5cm, minimum height=3cm,label={above:$D_2$}] {};]] {} ;

	 \draw[->,bend left=20,thick,shorten <=2pt,shorten >=2pt] (x1) to node[midway,fill=white,label=above:(\textnormal{\theoremautorefname}~\ref{thm:sufficient-condition-rapidity}: \ref{i:d1d2})] {$f$} (fx1);
  \draw[->,thick,shorten <=2pt,shorten >=2pt] (fx1) to node[midway,fill=white,solid] {$ALG_2^{rm}$} (a2prime);
  \draw[->,thick,shorten <=2pt,shorten >=2pt] (x2) to node[midway,fill=white,solid] {$ALG_2$} (a2);
	 \draw[->,bend left=20,thick,shorten <=2pt,shorten >=2pt] (a2prime) to node[midway,fill=white,solid,label=below:(\textnormal{\theoremautorefname}~\ref{thm:sufficient-condition-rapidity}: \ref{i:d2d1})] {$g$} (a1);

 \end{tikzpicture}
 } 
	\caption{
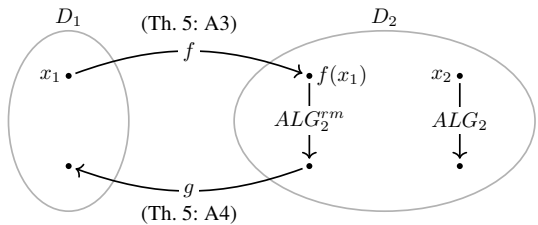
		Visualization of the proof of \cref{thm:sufficient-condition-rapidity} in case $OP$ is a transformation operation: %
Given runtime monotonic algorithm $ALG_2^{rm}$ implementing $OP$ on language $D_2$, the composed algorithm $ALG_1 \defn g \circ ALG_2^{rm} \circ f$ for $OP(D_1)$ is at least as rapid as $ALG_2$.
	To prove this, we consider $x_2\in D_2$ and an equivalent and at most only polynomially larger than $x_1\in D_1$ and show that $ALG_1$ takes at most polynomially more time on $x_1$ than $ALG_2$ on $x_2$.
	$ALG_1$ is also runtime monotonic.
	Horizontally-aligned instances of data structures are equivalent, i.e. represent the same quantum state.
	\supp{\cref{sec:proofs-of-sec-rapidity} provides a proof.}
	}
		\label{fig:rapidity-proof}
\end{restatable}

\begin{restatable}[A sufficient condition for rapidity]{theorem}{sufficientconditiontheorem}
	\label{thm:sufficient-condition-rapidity}
	Let $D_1,D_2$ be data structures with $D_1 \atleastassuccinctas D_2$ and $OP$ a $c$-ary operation.
	Suppose that,
	\begin{enumerate}[label=A\arabic*,leftmargin=*]
		\item $OP(D_2)$ requires time $\Omega(m)$ where $m$ is the sum of the sizes of the operands; and 
			\label{i:omega}
		\item for each algorithm $ALG$ implementing $OP(D_2)$, there is a runtime monotonic algorithm $ALG^{rm}$, implementing the same operation $OP(D_2)$, which is at least as rapid as $ALG$; and
			\label{i:rm}
		\item there exists a transformation from $D_1$ to $D_2$ which is (i)  weakly minimizing and (ii) runs in time polynomial in the output size (i.e, in time $\poly( \sizeof{\psi})$ for transformation output $\psi \in D_2$); and
			\label{i:d1d2}
		\item if $OP$ is a manipulation operation (as opposed to a query), then there also exists a polynomial time transformation from $D_2$ to $D_1$ (polynomial time in the input size, i.e, in $\sizeof \rho$ for transformation input $\rho \in D_2$).
			\label{i:d2d1}
	\end{enumerate}
	Then $D_1$ is at least as rapid as $D_2$ for operation $OP$.
\end{restatable}

\cref{fig:rapidity-proof} %
provides a proof outline that illustrates the need for (polynomial-time) transformations in order to execute operations on states represented in data structure $D_1$ on their $D_2$ counterpart.
The operation on $f(x_1)$ is at most polynomially slower than on its counterpart $x_2$ because $f$ produces a small instance (because $f$ is weakly minimizing), and because $ALG^{rm}_2$ is not much slower on such instances (because it is runtime monotonic).
We opted for weakly minimizing transformations (rather than \emph{strictly} minimizing transformations), because a minimum structure might be hard to compute and is not needed in the proof.
Runtime monotonic algorithms are ubiquitous, for instance, most operations on MPS scale polynomially in the bond dimension and number of qubits~\cite{vidal2003efficient}.
Finally, we emphasize that for canonical data structures $D$, each algorithm is runtime monotonic and any transformation $D_1 \rightarrow D$ is weakly minimizing.

\subsection{Rapidity between Quantum Representations}
\label{sec:transformations}
We now capitalize on the sufficient condition of \cref{thm:sufficient-condition-rapidity} by revealing the rapidity relations between data structures for all operations satisfying \ref{i:omega} and \ref{i:rm}.
\cref{thm:rapidity} shows our findings.
We highlight the result that MPS is at least as rapid as \qmdd in \cref{thm:mps-as-rapid-as-qmdd}. Its proof provides the required transformations from MPS to \qmdd and back.

\begin{restatable}{theorem}{thmrapidity}
\label{thm:rapidity}
	The rapidity relations in \cref{fig:rapidity} hold. 
\end{restatable}
\begin{restatable}{theorem}{mpsasrapidasqmdd}
	MPS is at least as rapid as \qmdd for all operations
	satisfying \ref{i:omega}~and~\ref{i:rm}.
	\label{thm:mps-as-rapid-as-qmdd}
\end{restatable}
\begin{proof}[Proof sketch]
	Since \qmdd is canonical, the runtime monotonicty (\ref{i:rm} of \cref{thm:sufficient-condition-rapidity}, see \cref{def:runtime-monotonic}) and weakly-minimizing (\ref{i:d1d2}) requirements are satisfied automatically.
	Hence we only need to provide efficient transformations in both directions (\ref{i:d1d2}-\ref{i:d2d1}): We transform \qmdd to MPS by choosing its matrices $A^x_k$ to be the weighted bipartite adjacency matrices of the $x$-edges (low / high edges) between level $k$ and $k-1$ nodes of the \qmdd.
	In the other direction, an MPS can be efficiently transformed into an \qmdd by contracting the first open index with $\ket 0$ and then $\ket 1$ to find the coefficients in the $\ket 0$ and $\ket 1$ parts of the state's Shannon decomposition; and repeating this recursively for all possible partial assignments.
	Dynamic programming using the fidelity operation (efficient for MPS) ensures that only a polynomial number of recursive calls are made.
\end{proof}

The relations not involving MPS involve QDDs.
QDDs are canonical data structures as explained in \cref{sec:preliminaries}, so that runtime monotonicity and weak minimization of transformations are automatically ensured.
\supp{In \cref{sec:transformations-between-QDDs}, we give detailed transformations between QDDs.}
Transformations between different QDDs can be realized using the well-known \textsc{MakeNode} procedure (see \cref{sec:preliminaries}), in linear time in the resulting QDD size (using dynamic programming).

\begin{restatable}[h!]{figure}{figrapidity}

\centering
\scalebox{.7}{
	\begin{tikzpicture}[node distance=.5cm,minimum height=.5cm]

		\node[draw] (mps) 			   {MPS};
		\node[draw, left = 1cm of mps] (limdd) {\limdd};
		\node[draw,  below = of limdd] (qmdd) {\qmdd};
		\node[draw, below = of qmdd] (add) {ADD};

		\draw[\bigarrowhead, solid] (qmdd.north)     to (limdd.south);
		\draw[\bigarrowhead, solid] (add.north)      to (qmdd.south);
		\draw[\bigarrowhead, solid, bend right = 20] (qmdd.east) to (mps.south);
	\end{tikzpicture}
}
	\caption{Rapidity relations between data structures considered here.
	A solid arrow $D_1\to D_2$ means $D_2$ is at least as rapid as $D_1$ 
	for all operations satisfying \ref{i:omega}~and~\ref{i:rm}~of~\cref{thm:sufficient-condition-rapidity}.
	}
	\label{fig:rapidity}

\end{restatable}

\section{Related work}
\label{sec:related-work}

\citeauthor{darwiche2002knowledge} pioneered the knowledge compilation map approach followed here.
\citeauthor{fargier2014knowledge} employ the same method  to compare decision diagrams for real-valued functions.
In order to mend a deficiency in the notion of tractability, \citeauthor{lai2017new} contributed the notion of \concept{rapidity},
which we extend in \cref{sec:rapidity}

Alternative ways to deal with non-canonical data structures includes finding a canonical variant, as has been done for MPS~\cite{perez2007matrix}, although it is not canonical for all states. For other structures, such as RBM, d-DNNF and CCDD (mentioned before) such canonical versions do not exist as far as we know, and might not be tractable.

Often the differences between DD variants in \cref{tab:dds} are subtle differences in the constraints used for obtaining canonicity. While perhaps irrelevant for the (asymptotic) analysis, we emphasize that these definitions have great practical relevance. For instance, canonicity makes dynamic programming efficient, which enables fast implementations of QDD operations, while it also reduces the diagram size. Moreover, in practice the behavior of floating point calculations interacts with the canonicity constraints~\cite{zulehner2019efficiently}, significantly affecting performance. Like \cite{fargier2014knowledge}, we do not analyze numerical stability and focus on asymptotic behavior.

The affine algebraic decision diagram (AADD), introduced by \citeauthor{sanner2005AffineADDs}, augments the \qmdd as shown in \cref{tab:dds}.
Their work proves that the worst-case time and space performance of AADD is within a multiplicative factor of that of \add. The concept of rapidity makes explicit that this is the case only when equivalent inputs are considered for both structures.
We note that AADD would be able to represent \sumstate (a hard case for all QDDs studied here), so its succinctness relation with respect to RBM is still open.
We omit Affine \add (AADD)~\cite{sanner2005AffineADDs}, since to the best of our knowledge these have not been applied to quantum computing yet.
The \limdd data structure is implemented and compared against \qmdd in \cite{limdd2}.

\citeauthor{burgholzer2023tensor} compare tensor networks, including MPS, to decision diagrams, on slightly different criteria, such as abstraction level and ease of distributing the computational workload on a supercomputer.
Hong et al. introduce the Tensor Decision Diagram, which extends the \qmdd so that it is able to represent tensors \cite{hong2020tensor} and their contraction.
Context-Free-Language Ordered Binary Decision Diagrams (CFLOBDDs)~\cite{sistla2023cflobdds,sistla2023weighted,quasimodo} achieve a similar goal by extending BDDs with insights from visibly pushdown automata~\cite{vpda}.

The stabilizer formalism is a tractable but non-universal method to represent and manipulate quantum states \cite{aaronson2008improved}.
Its universal counterpart is the stabilizer decomposition-based
method \cite{bravyi2016trading}, which relies on a linear combination of representations in the stabilizer formalism that is exponential only in the number of $T$ gates in the circuit. Although this is a highly versatile technique~\cite{bravyi2017improved}, no superpolynomial lower bounds on their size are known at the moment, making it less useful to include them in a knowledge compilation map.

Similarly, we expect the relation between tree tensor networks~\cite{orus2014practical} and SDDs~\cite{darwiche2011sdd} to be similar the relation between \qmdd and MPS, since SDD also extend the linear QDD variable order with a `variable tree.'

\citeauthor{ablayev2001computationalQOBDD} introduce the Quantum Branching Program, a branching program whose state is a superposition of the nodes on a given level, some of which are labeled \emph{accepting}.
An input string $x\in \{0,1\}^n$ determines how the superposition evolves.
A string $x\in \{0,1\}^n$ is said to be accepted by the automaton if, after evolving according to $x$, a measurement causes the superposition to collapse to an accepting state with probability greater than $\frac 12$.
Thus, the automaton accepts a finite language $L\subseteq\{0,1\}^n$.
Since this diagram accepts a language, rather than represents and manipulates a quantum state, we cannot properly fit it into the present knowledge compilation map.

\section{Discussion}
\label{sec:discussion}

We have compared several classical data structures for representing quantum information on their succinctness, tractability of relevant operations and rapidity.
We catalogued all relevant operations required to implement simulation of quantum circuits, variational quantum algorithms (e.g., quantum machine learning) and quantum circuit verification.
We have followed the approach of \cite{darwiche2002knowledge} and \cite{fargier2014knowledge} to map the succinctness and tractability of the data structures.
In order to mend a deficiency in the notion of tractability that was noticed by several researchers, we additionally adopt and develop the framework of \concept{rapidity}.
We contribute a general-purpose method by which data structures, whether canonical or not, may be analytically compared on their rapidity.

Common knowledge says that there is a trade-off between the succinctness, and the speed of a data structure.
In contrast, we find that, the more succinct data structures are often also more rapid.
For example, we find that algebraic decision diagrams, \qmdd, and matrix product states are each successively more succinct and more rapid.
However, in practice, the \qmdd has achieved striking performance on realistic benchmarks, \cite{zulehner2018advanced} due to a successful sustained effort to optimize the software implementation. But we emphasize that e.g. MPS was not developed with the intention for circuit simulation but for quantum simulation, and \qmdd is vice versa.
Therefore, more research is needed to compare the relative performance of these data structures on the various application domains in practice.

In future work, we could also consider unbounded operations (multiple applications of the same gates) and other data structures, e.g., tree tensor networks \cite{orus2014practical} and affine algebraic decision diagrams \cite{sanner2005AffineADDs}.

\bibliography{lit}

{\appendix}
{\onecolumn}
{\section{Additional Preliminaries on Quantum Information and Decision Diagrams}
\label{sec:prelims2}

\subsection{Brief Introduction to Quantum Information}
\label{sec:quantum-intro}

This appendix presents additional background on quantum information necessary to understand the other appendices. For a detailed treatment, we refer to \citeauthor{nielsen2000quantum}.

The basic unit of quantum information is a quantum bit or \concept{qubit}.
The joint state $\phi$ of $n$ qubits is described by a unit vector of $2^n$ complex numbers, denoted $|\phi\rangle$ in Dirac notation.
Equivalently, a state $\phi$ is described by an \concept{amplitude function} $f_{\phi} \colon \{0,1\}^n\to\mathbb C$ satisfying $\sum_{\vec x\in \{0,1\}^n}|f_{\phi}(\vec x)|^2=1$.
Here a complex number $z = a + bi$ with $a, b \in \mathbb{R}$ has squared modulus $|z|^2 = z^* \cdot z = a^2 + b^2$ and complex conjugate $z^* = a - bi$.
The computational basis of the $n$-qubit vector space $\mathbb{C}^{2^n}$ is $\ket{k} = (0, 0, \dots, 0, 1, 0, \dots, 0)^T$ for $k\in \{0, 1, \dots, 2^n-1\}$, where the single entry $1$ occurs at the $k$-th position (often we write $k$ in binary, i.e. $k \in \{0, 1\}^n$) and $(.)^T$ denotes vector transposition.
Thus $\ket{\phi} = \sum_{x \in \{0, 1\}^n} f_{\phi}(x) \ket{x}$.
We also write $\ket{k} = \ket{x_1x_2\dots x_n} = \ket{x_1} \otimes \ket{x_2} \otimes \dots \otimes \ket{x_n}$ for $x_j \in \{0,1\}$ and $x$ the binary representation of $k \in \mathbb{N}_{\geq 0}$.
Here $A \otimes B$ for a $k \times l$ matrix $A$ and an $m \times n$ matrix $B$ is the $km \times ln$ tensor product:
\[
A\otimes B = 
\begin{pmatrix}
A_{00} B & A_{01} B & \dots & A_{0l} B\\
A_{10} B & A_{11} B & \dots & A_{1l} B\\
\vdots & \vdots & \ddots & \vdots\\
A_{k 0 } B & A_{k 1} B & \dots & A_{k l} B\\
\end{pmatrix}.
\]
Two separate quantum states $\ket{\phi_A}, \ket{\phi_B}$ on $n_{A}, n_{B}$ qubits have joint state $\ket{\phi_A}\otimes\ket{\phi_B}$.
If a quantum state $\ket\phi$ on $n=n_{A} + n_{B}$ qubits is decomposable as $\ket{\phi}=\ket{\phi_A}\otimes \ket{\phi_B}$, it is \emph{separable} (with respect to bipartition $A,B$) and \emph{entangled} otherwise.
We say that a complex vector $\ket{\phi}\in\mathbb C^{2^n}$ is \emph{normalized} when $\ket{\phi}$ is a unit vector.

Two quantum states $\phi,\psi$ are \concept{equivalent} iff there exists a complex value $\lambda\in\mathbb C$ such that $\ket\phi=\lambda\cdot\ket\psi$.
Given two $n$-qubit states $\phi, \psi$, their inner product is $\langle \phi | \psi \rangle \defn \sum_{\vec  x\in \{0,1\}^n} f_{\phi}(\vec x)^* \cdot f_{\psi}(x)$, and their fidelity is $|\langle \phi | \psi \rangle|^2$.
Equivalently, the inner product is also (the $1\times 1$-matrix entry) $\bra{\phi} \cdot \ket{\phi}$ where $\bra{\phi} \defn \left( \ket{\phi} \right)^{\dagger}$ is a column vector where $(.)^{\dagger}$ is the complex transpose operator, which first takes a matrix's transpose followed by entry-wise complex conjugation.
Note that if $k \in \set{0, 1}^n$, the $k$-th entry in the state vector $\ket \phi$ is $\braket{k|\phi} = f_\phi(k)$.
\begin{example}
\cref{fig:ghz-examples} gives an example of a quantum state on $3$ qubits, the so-called GHZ state: $\ket{\text{GHZ}}=\nicefrac{1}{\sqrt 2}(\ket{000}+\ket{111})$.
This state is displayed as state vector (left), followed by representations in all data structures used in this paper.
\end{example}

A quantum state can be manipulated by applying a quantum \concept{gate} or a measurement to it.
An $n$-qubit gate $U$ is a unitary transformation $U\colon \mathbb C^{2^n}\to\mathbb C^{2^n}$, i.e., it satisfies $U^\dagger U = U U^{\dagger} = \id[2^n]$ where $\id[2^n]$ is the $n$-qubit identity operator mapping each vector to itself.
Examples of gates are the four Pauli matrices defined in \cref{sec:data-structures-definitions}, the controlled-$U$ gate mapping $\ket{0}\otimes \ket{\phi}$ to itself and $\ket{1}\otimes\ket{\phi}$ to $\ket{1}\otimes U\ket{\phi}$, the Hadamard gate $H = \nicefrac{1}{\sqrt{2}} \begin{smallmat} 1 & 1 \\ 1 & -1\end{smallmat}$, the gate $T = \begin{smallmat} 1 & 0\\ 0 & 
e^{\nicefrac 14\pi}
\end{smallmat}$ and the phase gate $S = T^2$.
A single-qubit computational-basis measurement, performed on the $k$-th qubit of an $n$-qubit state $\phi$, yields outcome $m\in \{0, 1\}$ with probability $p_m \defn \sum_{x\in \{0, 1\}^{n-1}} |f_{\phi}(x, x_k = m)|^2$, and the post-measurement state $\psi$ is found by setting a vector entries with $x_k \neq m$ to 0 and subsequent rescaling; i.e., $f_{\psi} (x) := f_{\phi}(x) / p_m$ if $x_k = m$ and $f_{\psi}(x) := 0$ otherwise.

\begin{example}
	A small quantum circuit is shown in \cref{fig:ghz-circuit}.
The initial state is $\ket{0}\otimes\ket{0}\otimes \ket{0}$.
First, a Hadamard gate is applied to the first qubit, after which the state becomes $\nicefrac{1}{\sqrt 2}(\ket{000}+\ket{100})$.
Then a controlled-$X$ gate is applied to the first and second qubit.
Here the first qubit acts as the \emph{control}; the second qubit is the \emph{target}.
The state of the quantum system is now $\nicefrac{1}{\sqrt 2}(\ket{000}+\ket{110})$.
The state before measurement is called the GHZ state, $\ket{\text{GHZ}}=\nicefrac{1}{\sqrt 2}(\ket{000}+\ket{111})$.
Finally, a measurement is applied to the third qubit.
After this measurement, the state ``collapses'' to either state $\ket{000}$ or $\ket{111}$ with equal probability.

\begin{figure}[h]
\begin{align*}
	\Qcircuit @C=2.3em @R=0.7em {
	& \ket{0} & & \gate{H} & \ctrl{1} & \qw      & \qw &\qw \\
	& \ket{0} & & \qw      & \gate{X} & \ctrl{1} & \qw &\qw \\
	& \ket{0} & & \qw      & \qw      & \gate{X} & \meter
	}
\end{align*}
\caption{$3$-qubit quantum circuit preparing the~GHZ~state.
}
\label{fig:ghz-circuit}
\end{figure}
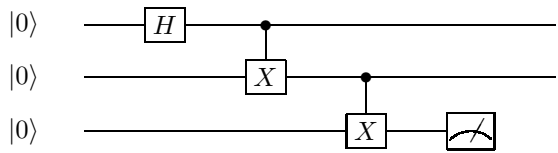
\end{example}

We now formalize the application domains from \cref{sec:application-domains}.
In \concept{circuit simulation}, a quantum circuit consists of (i) an initial quantum state, (ii) several quantum gates  followed by (iii) measuring one or more quantum bits.\footnote{In general, these steps can be intertwined and include adaptive operations, i.e. which depend on measurement outcomes, with (i) replaced by resetting the quantum state.}
Simulating a circuit using the data structures in this work starts by constructing a data structure for (i), followed by manipulating the data structure by applying the gates and measurements of (ii) one by one.
The structure supports \emph{strong simulation} if it can produce the probability function; or \emph{weak simulation} if it merely produces a sample, for each measurement in~(iii).
Next, the quantum circuit of \concept{variational methods} consist of a quantum circuit; then the output state $\ket{\phi}$ is used to compute $\langle \phi | O | \phi\rangle$ for some linear operator $O$ (an \concept{observable}) which is hermitian ($O = O^{\dagger}$), for example $O$ is the Hamiltonian. This computation reduces to $\langle \phi | \psi \rangle$ with $\ket{\psi} := O\ket{\phi}$ which could be any complex vector.
Often one considers \concept{local} observables.
Specifically, an $n$-qubit observable $O$ is called $k$-local when it can be written as $O=P\otimes \id[2^{n-k}]$ where $P$ is a $k$-qubit observable.
Last, \concept{circuit verification} relies on checking approximate or exact equivalence of quantum states.
This extends to unitaries, which all data structures from this paper can represent also, but which we will not treat for simplicity.

\subsection{Preliminaries on decision diagrams}
\label{sec:ddprelims}

Here, we present some additional background required to understand the results in the other appendices. For a detailed treatment of decision diagrams, we refer to~\cite{somenzi1999binary}.

The tractability of the QDD algorithms heavily relies on two properties: canonicity, i.e., there exists a unique decision diagram for each quantum state, and dynamic programming, i.e., avoiding unnecessary recursion by storing intermediary results in a cache.
Canonicity implies that for each quantum state $\phi$ there is a unique diagram $x \in \allformulae{D}{\phi}$, and moreover no other (non-canonical) data structure $y \in \allformulae{D}{\phi}$ has fewer nodes or edges than $x$.
All QDDs in this work can be made canonical by requiring that the QDD satisfies a set of reduction rules  (for \qmdd see \cite{miller2006qmdd} and for \limdd see \cite{vinkhuijzen2021limdd}) and all isomorphic nodes are subsequently merged (see \cref{tab:dds}).
This so-called \concept{reduced} QDD can be obtained in linear time using a standard \textsc{MakeNode} procedure (see also \cref{sec:transformations}). QDD manipulation algorithms also use the \textsc{MakeNode} procedure to ensure that the result is reduced efficiently.
For these reasons, we may always assume that QDDs are reduced.

Let $f\colon \{0,1\}^n\to\mathbb C$ be a function, and $y_n,y_{n-1},\ldots, y_k\in\{0,1\}$ a partial assignment.
Then $f_{y}\colon \{0,1\}^{k-1}\to\mathbb C$ is the \emph{subfunction of} $f$ \emph{induced by} $y$, defined as $f_y(x)=f(y,x)$.
Suppose a QDD has root node $v$ representing a pseudo-Boolean function $f\colon \{0, 1\}^n \rightarrow \mathbb{C}$.
For any $y_n,\ldots, y_k$, this diagram contains a node $w$ which represents the induced subfunction $f_y\colon \{0,1\}^{k-1}\to\mathbb C$, or which represents a function which is isomorphic to $f_y$ under $\mathcal {E}_k$.
The vertices of the QDD are divided into \emph{layers}, i.e., a node which is $k$ edges away from a leaf is said to be in layer $k$, and each edge points from a vertex in layer $k$ to a node in layer $k-1$.
Each node in layer $k$ represents a pseudo-Boolean function of the form $f\colon \{0,1\}^{k}\to\mathbb C$.

We will use the following proposition as a general intuition in many of the following proofs.

\begin{proposition}
A reduced QDD representing pseudo-Boolean function (quantum state) $f$ has as many nodes on level $\ell$, as there are unique induced subfunctions $f_{\vec a}$ for $\vec a \in \set{0,1} ^ {n - \ell}$, modulo isomorphism under $\mathcal E_k$ (see \cref{def:iso}).
\end{proposition}

\newcommand\weight[1]{\mathit{weight}(#1)}

}
{\section{Proofs of \cref{sec:succinctness-separations}}
\label{sec:proofs-of-sec-succinctness}

In this appendix, we prove \cref{thm:succinctness} as reproduced below. We also reproduce \cref{fig:succinct} in \cref{fig:succinct2}, which additionally includes references to the respective lemmas. \cref{sec:prelims2} contains relevant preliminaries on quantum information and QDDs.

\renewcommand\supp[1]{#1} %
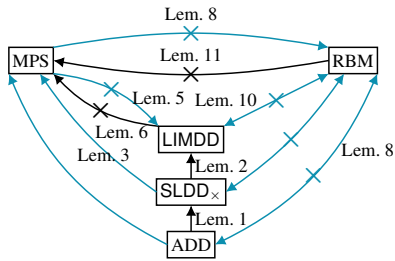
\begin{figure}[b!]
\centering
\vspace{-1em}
\scalebox{.7}{
\begin{tikzpicture}[node distance=.7cm,minimum height=.5cm]
		\node[draw] (mps) 			   {MPS};
		\node[draw, right =2cm of mps, yshift=-1.5cm] (limdd) {\limdd};
		\node[draw, right =2cm of limdd, yshift=1.5cm] (rbm) {RBM};
		\node[draw, below = .5cm of limdd] (qmdd) {\qmdd};
		\node[draw, below = .5cm of qmdd] (add) {ADD};

		\draw[\bigarrowhead, bend left=20] (limdd.north west) to node[midway] {$\boldsymbol{\bigtimes}$}
											node[pos=0.3,left=1cm,below=0cm] {\textcolor{black}{\supp{\cref{thm:succ-limdd-vs-mps}}}}
											(mps.south east);
		\draw[\bigarrowhead, bend left=20, color=light-blue] (mps.south east) to node[midway] {$\boldsymbol{\bigtimes}$}
											node[pos=1,right=0cm,above=0.3cm] {\textcolor{black}{\supp{\cref{thm:succ-mps-vs-limdd}}}}
											(limdd.north west);

		\draw[\bigarrowhead, bend left=10, color=light-blue] (qmdd.west) to node[pos=.3,below left=-.5cm]
									{\textcolor{black}{\supp{\cref{thm:mps-exp-more-succinct-than-qmdd}}}} (mps);

		\draw[\bigarrowhead, bend left=20, color=light-blue] (add.west) to
													 (mps.south west);

		\draw[\bigarrowhead, bend left=10] (rbm.west) to node[midway] {$\boldsymbol{\bigtimes}$} 
													node[midway,above=0.125cm] {\supp{\cref{thm:succ-rbm-vs-mps}}}
											(mps.east);
		\draw[\bigarrowhead, bend left=10, color=light-blue] (mps.north east) to node[midway] {$\boldsymbol{\bigtimes}$} 
													node[midway,above=0.125cm] {\textcolor{black}{\supp{\cref{thm:add-rbm}}}}
											(rbm.north west);

		\draw[\bigarrowheadb, bend right=0, color=light-blue] (rbm.south west) to node[midway] {$\boldsymbol{\bigtimes}$}
											node[pos=.7,above left=-1mm] {\textcolor{black}{\supp{\cref{thm:succ-rbm-vs-limdd}}}}
											(limdd.north east);

		\draw[\bigarrowheadb, bend left=10, color=light-blue] (rbm) to node[midway] {$\boldsymbol{\bigtimes}$}
										   (qmdd.east);

		\draw[\bigarrowheadb, bend left=20, color=light-blue] (rbm.south east) to node[midway] {$\boldsymbol{\bigtimes}$}
											node[pos=.3,below right=-1mm] {\textcolor{black}{\supp{\cref{thm:add-rbm}}}} 
											 (add.east);

		\draw[\bigarrowhead, bend right=0] (qmdd) to  node[pos=.5,right=.01cm] 
									{\supp{\cref{thm:succ-limdd-more-qmdd}}}	 (limdd);
		\draw[\bigarrowhead, bend right=0] (add) to node[pos=.5,right=.01cm]
									{\supp{\cref{thm:succ-qmdd-more-add}}} 	(qmdd);

	\end{tikzpicture}
}
\caption{Succinctness relations between various classical data structures for representing quantum states.
Solid arrows $A\to B$ denote $B\strictlymoresuccinctthan A$, i.e., $B$ is strictly more succinct than~$A$.
Crossed arrows $A\crossedto B$ denote a separation $B \notmoresuccinctthan A$; a bidirectional crossed arrow implies incomparability.
Blue arrows indicate novel relations that we identified.
}
\label{fig:succinct2}
\end{figure}

\succinctness*
\begin{proof}
The proofs for individual relations are stated in the lemmas referenced by \cref{fig:succinct2}, which is otherwise equal to \cref{fig:succinct}.

Note that we do not include a proof for every arrow (direction), since several can be derived through transitivity properties. All unlabeled edge (directions) can be derived as follows:
\begin{itemize}
\item  MPS $\strictlymoresuccinctthan$ \add follows from MPS $\strictlymoresuccinctthan$ \qmdd and \qmdd $\strictlymoresuccinctthan$ \add
\item  \limdd $\strictlymoresuccinctthan$ \add follows from \limdd $\strictlymoresuccinctthan$ \qmdd and \qmdd $\strictlymoresuccinctthan$ \add
\item \qmdd $\notmoresuccinctthan$ RBM follows from  \limdd  $\notmoresuccinctthan$ RBM and \limdd $\strictlymoresuccinctthan$ \qmdd
\item \add $\notmoresuccinctthan$ RBM follows from  \limdd  $\notmoresuccinctthan$ RBM and \limdd $\strictlymoresuccinctthan$ \add
\item RBM  $\notmoresuccinctthan$ MPS follows from  RBM  $\notmoresuccinctthan$ \add and MPS $\strictlymoresuccinctthan$ \add
\item RBM  $\notmoresuccinctthan$ \qmdd follows from  RBM  $\notmoresuccinctthan$ \add and \qmdd $\strictlymoresuccinctthan$ \add
\item RBM  $\notmoresuccinctthan$ \limdd follows from  RBM  $\notmoresuccinctthan$ \add and \limdd $\strictlymoresuccinctthan$ \add
\end{itemize}
This completes the proof of all stated succinctness relations.
\end{proof}

\begin{lemma}
	\label{thm:succ-qmdd-more-add}
	\qmdd is exponentially more succinct than \add.
\end{lemma}
\begin{proof}
Since \add is a special case of \qmdd (\cref{sec:data-structures-definitions}), \qmdd is at least as succinct.

 \citeauthor{fargier2013semiring} prove an exponential separation in Prop.~10.
 The proposition itself only mentions a  superpolynomial separation; the fact that the separation is in fact exponential is contained in the proof.
\end{proof}

\begin{lemma}
	\label{thm:succ-limdd-more-qmdd}
	\limdd is exponentially more succinct than \qmdd.
\end{lemma}
\begin{proof}
Since \qmdd is a special case of \limdd (\cref{sec:data-structures-definitions}), \limdd is at least as succinct.

\citeauthor{vinkhuijzen2021limdd} show an exponential separation for so-called `cluster states.'
\end{proof}

\begin{lemma}
\label{thm:mps-exp-more-succinct-than-qmdd}
	MPS is exponentially more succinct than \qmdd.
\end{lemma}
\begin{proof}
We show in \cref{sec:qmdd-to-mps} that MPS is at least as succinct as \qmdd, by showing that every \qmdd can be translated to MPS in linear time.

We  provide a state $\ket\phi$ on $n$ qubits, which has an exponential-sized \qmdd, but a polynomial-sized MPS.
	Let $(x)_2\in\mathbb Z$ be the integer represented by a bit-string $x\in \{0, 1\}^n$.
	The state of interest is
	\begin{align}
	\ket{\phi} = \sum_{x\in\{0,1\}^n} (x)_2\ket x = \sum_{x\in \{0,1\}^n}\left(\sum_{j=1}^n 2^{j-1}x_j\right)\ket x
	\end{align}
	\citeauthor{fargier2013semiring} show that this state has exponential-sized \qmdd (Prop.~10). 
	On the other hand, it can be efficiently represented by the following MPS of bond dimension $2$:
	\begin{align} 
	A_n^0 =& \begin{smallmat}1 & 0 \end{smallmat}
&  A_j^0 =& \begin{smallmat}1 & 0 \\ 0 & 1\end{smallmat}
&	A_1^0 = & \begin{smallmat}0 \\ 1\end{smallmat}  \\
	A_n^1 =& \begin{smallmat}1 & 2^{n-1} \end{smallmat}
&  A_j^1 =& \begin{smallmat}1 & 2^{j-1} \\ 0 & 1 \end{smallmat}
&  A_1^1 = & \begin{smallmat}1 \\ 1\end{smallmat} 
	\end{align}
	Here $j$ ranges from $2\ldots n-1$.
To show this, we can write
\[
A_n^{x_n} = \begin{bmatrix}1 & x_n \cdot 2^{n-1} \end{bmatrix}
\qquad
A_j^{x_j} = 
\begin{bmatrix} 1& x_j \cdot 2^{j-1}\\ 0 & 1 \end{bmatrix} \textnormal{\mbox{ } for $j=2, \dots, n-1$}
\qquad
A_1^{x_1} = \begin{bmatrix} x_1 \\ 1 \end{bmatrix}
\]
Hence we can write
\[
A_n^{x_n}\cdot  \cdots A_1^{x_1}
=
\begin{bmatrix}1 & x_n \cdot 2^{n-1}\end{bmatrix}
\cdot
\begin{bmatrix} 1& \sum_{j=2}^{n-1} x_j \cdot 2^{j-1}\\ 0 & 1 \end{bmatrix} 
\cdot
\begin{bmatrix} x_1 \\ 1 \end{bmatrix}
=
\begin{bmatrix} 1 & x_n \cdot 2^{n-1} \end{bmatrix}
\cdot
\begin{bmatrix} \sum_{j=1}^{n-1} x_j \cdot 2^{j-1}\\  1 \end{bmatrix} 
=
\begin{bmatrix} \sum_{j=1}^{n}  2^{j-1} \cdot x_j \end{bmatrix}
\]
\end{proof}

The following quantum state, called \sumstate, will feature in several of the below proofs.
Specifically, we will show that RBM and MPS can represent this state efficiently, whereas \limdds cannot.
A similar state will be used to show that \limdd does not support the Swap operation.  We omit normalization factors, as all data structures are oblivious to them.
\begin{align}
\sumstate = \ket{+}^{\otimes n} + \bigotimes_{j=1}^n(\ket 0+e^{i\pi 2^{-j-1}}\ket 1)
\end{align}
\begin{lemma}
	\label{thm:state-has-exp-limdd}
	The \limdd of \sumstate has size $2^{\Omega(n)}$ for every variable order.
\end{lemma}
\begin{proof}

We compute that the amplitude function for \sumstate is 
\begin{align}
f(\vec{x}) = 1 + e^{i\pi \sum_{j=1}^n x_j \cdot 2^{-j-1}}.\label{eq:sum}
\end{align}
We note that $f$ is injective and never zero,
and that $\frac{f(\vec x)}{f(\vec y)}$ is injective on the domain where $\vec x \neq \vec y$.

We now study the nodes $v$ at level $1$ (with $\nindex(v) = 1$) via the subfunctions they represent,
considering all variable orders.
These nodes represent subfunctions on one variable.
So we take out one variable $x_k\in \vec x = \set{x_1,\dots,x_n}$.
Without loss of generality, we may pick $x_1$ because the summation in \cref{eq:sum} is commutative. For each assignment $\vec a \in \set{0,1}^{n-1}$, we obtain the function:
\[
f_{\vec a}(x_1) = 1 + e^{i\pi \sum_{j=2}^n a_j \cdot 2^{-j-1}} \cdot
e^{i\pi \cdot \nicefrac14 x_1 }.
\]

We now show that for any $\vec a \neq \vec c \in \set{0,1}^{n-1}$ there is no $Q \in\paulilim_1$ such that $f_{\vec a} = Q f_{\vec c}$.

Let $Q = \alpha P$ for $\alpha \in  \mathbb C \setminus \set 0, P\in \set{\id, X,Y, Z}$, so $f_{\vec a} = \alpha P f_{\vec c}$.
Furthermore, define
$\alpha = \alpha(z, x, \vec{a}, x_1) = (-1)^z \cdot \frac{f_{\vec{a}}(x_1 \oplus x = 0)}{f_{\vec{a}}(x_1 \oplus x = 1)}$ for
 $P = X^x \cdot Z^z $ with $x, z\in \{0, 1\}$, absorbing the factor $i$ of $Y$ and -1 of $ZX$ in $\alpha$. The function $\alpha$ is injective, i.e., $\alpha(s) = \alpha(t)$ implies $s = t$, based on our earlier observations about $f$.

It follows that each subfunction $f_{\vec{a}}$ requires a separate node at level 1.
So there are $\Omega(2^{n-1})$ nodes.
\end{proof}

\begin{lemma}
	\label{thm:succ-mps-vs-limdd}
	There is a family of quantum states with polynomial-size MPS but exponential-size \limdd.
\end{lemma}
\begin{proof}
	MPS require only bond dimension $2$ to represent the state \sumstate.
	Namely, $\sumstate$ is the sum of two product states.
	Product states always have bond dimension $1$; and the sum of two states which have bond dimensions $d_1$ and $d_2$, has bond dimension at most $d_1+d_2$, as shown by \cref{thm:mps-addition}.
	
	Concretely, an MPS $A$ with bond dimension 2 for $\sumstate$ is given by 
	\begin{align}
	    A^{x_1}_1 = \begin{bmatrix} 1 & f_1(x_1)\end{bmatrix}, && A^{x_j}_j = \begin{bmatrix} 1 & 0 \\ 0 & \prod_{j=2}^{n-1} f_j (x_j)\end{bmatrix}\text{ for }j=2, 3, \dots, n-1, && A^{x_n}_n = \begin{bmatrix} 1 \\ f_n(x_n)\end{bmatrix},
	\end{align}	
	 where $x_1, x_2, \dots, x_n \in \{0, 1\}$ and $f_j$ is the function $\{0,1\}\rightarrow\mathbb{C}$ defined by$ f_j(0)=1$ and $f_j(1) = e^{i\pi 2^{-j-1}}$.
	 Some computation shows that, indeed, $A_1^{x_1} \cdot \dots \cdot A_n^{x_n} = \begin{pmatrix} 1 + \prod_{j=1}^n f_j (x_j) \end{pmatrix} = 1+e^{i\pi (x)_22^{-n-1}} = \braket{x|\text{Sum}}$, as desired.
	 Moreover, the bond dimension of the MPS $A$ is $2$.
	
	However, \cref{thm:state-has-exp-limdd} shows that \limdds require exponential size to represent the same state.
\end{proof}

\begin{lemma}
	\label{thm:succ-limdd-vs-mps}
	There is a family of quantum states with polynomial-size \limdd but exponential-size MPS.
\end{lemma}
\begin{proof}
	\limdd can efficiently represent any stabilizer state, but some stabilizer states require exponential-size MPS (in particular, the cluster state, among others \cite{vinkhuijzen2021limdd}.
\end{proof}

\begin{lemma}
	\label{thm:qmdd-to-mps}
	MPS is at least as succinct as \qmdd.	
\end{lemma}
\begin{proof}
	\cref{sec:qmdd-to-mps} provides a polynomial time transformation from \qmdd to MPS.
\end{proof}

\begin{figure}
\centering
\begin{tikzpicture}[
    scale=0.3,
    every path/.style={>=latex},
    every node/.style={},
    inner sep=0pt,
    minimum size=14pt,
    line width=1pt,
    node distance=.5cm,
    thick,
    font=\footnotesize
    ]

    \node[] (root)   {};
    \node[draw,circle, below =of root, xshift=0cm] (l1) {};
    \node[draw,circle, below =of l1] (l2) {};

    \node[draw,circle, below =of l1, xshift=1cm] (ma2) {};

    \node[draw,circle, below =of l2] (l3) {};
    \node[draw,circle, below =of l2, xshift=3cm] (r3) {};
    \node[draw,circle, below =of l3] (l4) {};
    \node[draw,circle, below =of r3] (r4) {};
    \node[draw,circle, below =of l3, xshift=1cm] (ma4) {};
    \node[draw,circle, below =of l3, xshift=2cm] (mb4) {};

    \node[draw,circle, below =.6cm of l4] (l5) {};
    \node[draw,circle, below =.6cm of r4] (r5) {};

    \node[draw,circle, below= of l5] (l6) {};
    \node[draw,circle, below= of r5] (r6) {};
    \node[draw,circle, below= of l5, xshift=1cm] (ma6) {};
    \node[draw,circle, below= of l5, xshift=2cm] (mb6) {};

	 \node[left of =l1] 	{$x_1$};
	 \node[left of =l2] 	{$x_2$};
	 \node[left of =l3] 	{$x_3$};
	 \node[left of =l4] 	{$x_4$};
	 \node[left  =.1cm of l5] 	{$x_{n-1}$};
	 \node[left of =l6] 	{$x_{n}$};

    \node[draw,circle,rectangle,minimum size=0.4cm, below=of l6] (leaf0) {$0$};
    \node[draw,circle,rectangle,minimum size=0.4cm, below=of r6] (leaf1) {$\frac{1}{A}$};

    \draw[<-] (l1) --++(90:2cm) node[pos=1.4,right] {};

    \draw[e0=0] (l1) edge  node[] {} (l2);
    \draw[e1=0] (l1) edge  node[] {} (ma2);

    \draw[e0=25] (l2) edge  node {} (l3);
    \draw[e1=25] (l2) edge  node {} (l3);

    \draw[e0=0] (ma2) edge  node {} (l3);
    \draw[e1=0] (ma2) edge  node {} (r3);

    \draw[e0=0] (l3) edge  node {} (l4);
    \draw[e1=0] (l3) edge  node {} (ma4);

    \draw[e1=0] (r3) edge  node {} (mb4);
    \draw[e0=0] (r3) edge  node[left] {} (r4);

    \node[ xshift=0.5cm, yshift=.5cm, below= of ma4] (a) {$\pmb{\vdots}$};

    \draw[e0=0] (l5) edge  node {} (l6);
    \draw[e1=0] (l5) edge  node {} (ma6);

    \draw[e0=0] (r5) edge  node {} (r6);
    \draw[e1=0] (r5) edge  node {} (mb6);

    \draw[e0=0] (ma6) edge  node {} (leaf0);
    \draw[e1=0] (ma6) edge  node {} (leaf1);

    \draw[e0=0] (mb6) edge  node {} (leaf1);
    \draw[e1=0] (mb6) edge  node {} (leaf0);

    \draw[e0=25] (l6) edge  node {} (leaf0);
    \draw[e1=25] (l6) edge  node {} (leaf0);

    \draw[e1=25] (r6) edge  node {} (leaf1);
    \draw[e0=25] (r6) edge  node {} (leaf1);
\end{tikzpicture}
\caption{An \add for the inner product function $IP'$ from \cref{thm:add-rbm} made up of stacked blocks, each consisting of a layer of 2 nodes and a layer of 4 nodes.
$A$ in the right leaf is the normalization constant from \cref{thm:add-rbm}.
}
\label{fig:add-ip}
\end{figure}
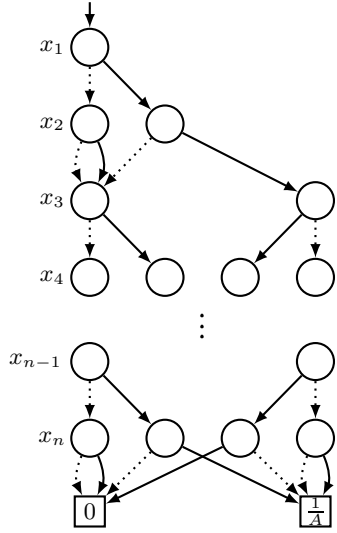

For proving the separation between RBM and \add, we use the seminal Boolean function $IP: \set{0, 1}^n \rightarrow \{0, 1\}, \vec{x} \mapsto \sum_{k=1}^{n/2} x_k x_{k + n/2} \mod 2$ for even $n$, which computes the inner product between the first half of the input with the second half.
\citeauthor{martens2013representational} show that any RBM requires a number of hidden weights $m$ which is necessarily exponential in $m$.

\begin{lemma}
\label{thm:add-rbm}
There is a quantum state that has linear representation both as \add, and \qmdd, and \limdd, and MPS, but requires exponential space when represented as RBM under any qubit order.
\end{lemma}
\begin{proof}
We will give the proof for the \add; the result will then follow for \qmdd, \limdd and MPS, since these are at least as succinct as \add.

Since we consider the representation size under any qubit variable order,
 we may as well interleave the order. That is, we 
  consider $IP'$ which equals $IP$ with $x_{k+1}$ and $x_{k + n/2}$ swapped, i.e. $IP'(x) = x_1 x_2 + x_3 x_4 + \dots + x_{n-1} x_n$.
Consider the $n$-qubit quantum state $\ket{\phi}$ where $\langle x | \phi\rangle = IP'(x) / A$ for $x\in \{0, 1\}$, where the normalization factor $A = \sqrt{\sum_{x\in \{0, 1\}} IP'(x)}$.
\citeauthor{martens2013representational} show that any RBM requires a number of hidden weights $m$ which is necessarily exponential in $m$.

There exists an \add which represents $\ket{\phi}$ in $O(n)$ space.
This \add is constructed from stacked blocks of two layers (of 2 and 4 nodes, respectively).
The $(k+1)/2$-th block (counting from 1 from the top) for odd $k = 1, 3, 5, \dots$ corresponds to computing the value $x_k \cdot x_{k+1}$ and adding it to the running value of $IP'(x_1, x_2, \dots, x_{k-1})$.
See \cref{fig:add-ip}.
\end{proof}

\begin{lemma}
	\label{thm:rbm-represents-sum}
	RBM can represent the state \sumstate with a single hidden node.
\end{lemma}
\begin{proof}
	All nodes have bias $0$, i.e., $\beta = [0]$ and $\alpha = [0, \dots, 0]^T$ (a length-$n$ vector).
	The weight on the edge between the hidden node and the $j$-th visible node is $e^{i\pi 2^{-j-1}}$.
	Then the RBM is defined by the multiplicative term of this hidden node, yielding
	\begin{align}
	\psi(\vec x) = 1 + e^{w\cdot \vec x} = 1 + \prod_{j=1}^ne^{x_ji\pi 2^{-j-1}}
	\end{align}
This corresponds exactly with the sum state: $\ket \psi = \sumstate$.
\end{proof}
 
\begin{lemma}
	\label{thm:succ-rbm-vs-limdd}
	There is a state with a RBM of size $\mathcal O(n)$ but which requires \limdd of size $2^{\Theta(n)}$, for every variable order.
\end{lemma}
\begin{proof}
	RBM can represent the state \sumstate, by \cref{thm:rbm-represents-sum}.
	However, \cref{thm:state-has-exp-limdd} shows that \limdds require exponential size to represent this state.
\end{proof}

\begin{lemma}
	\label{thm:succ-rbm-vs-mps}
	There is a family of states with polynomial-size RBM but exponential-size MPS.
\end{lemma}
\begin{proof}
	RBM can efficiently represent stabilizer states, as shown by \citeauthor{zhang2018efficient}.
	\citeauthor{vinkhuijzen2021limdd} show that some stabilizer states require exponential-size MPS (in particular, the cluster state, among others).
\end{proof}

}
{\section{Proofs of \cref{sec:tractable-queries}}
\label{sec:proofs-of-sec-tractability}

In this appendix, we prove \cref{thm:tractability} and  \cref{thm:fidelity} from \cref{sec:tractable-queries}.

\cref{thm:tractability} is restated below.
The proofs are organized per row of the table, so there is one section for each data structure.
\cref{sec:prelims2} contains relevant preliminaries on quantum information and QDDs.

\thmtractability*

\renewcommand\supp[1]{#1} %
\begin{table}[b!]

\label{tab:tractability2}
\end{table}

We restate the other main result \cref{thm:fidelity} here and provide a proof.

\thmfidelity*
\begin{proof}
	\cref{thm:even-subgraphs-ETH-hard} proves that \limdd does not admit a polynomial time algorithm unless the exponential time hypothesis fails. 
	\cref{thm:rbm-fidelity-ETH-hard} concludes the same for RBM.
\end{proof}

\subsection{Easy and hard operations for \add}
\label{app:trac:add}

As noted in \cref{sec:data-structures-definitions}, the decision diagrams are special cases of each other. In particular, \add specializes \qmdd, which specializes \limdd. From this, it immediately follows that $\limdd  \atleastassuccinctas \qmdd \atleastassuccinctas \add$. We also use this fact in the below proofs.

\begin{lemma}
	\add supports \samp and \pro.
\label{th:add-samp}
\label{th:add-prob}
\label{th:add-eq}
\end{lemma}
\begin{proof}
Since \add specializes \limdd (see \cref{thm:limdd-prob}).	
\end{proof}

\begin{lemma}
\add  supports inner-product $\innerprod \phi \psi$.
\label{th:add-inprod}
\end{lemma}
\begin{proof}
Since \add is a specialization of \qmdd (see \cref{th:qmdd-inprod}).
\end{proof}

\begin{lemma}
\add  supports \addi and \eq.
\label{th:add-add}
\label{th:add-eq}
\end{lemma}
\begin{proof}
See \citeauthor{fargier2014knowledge} Table 1 (EQ) and Table 2 (\textbf{+BC}).
\end{proof}

\begin{lemma}
\add supports \loc, and hence also \had, \xyz, $T$, \swap and \cz.
\label{th:add-loc}
\label{th:add-xyz}
\label{th:add-had}
\label{th:add-t}
\label{th:add-swap}
\label{th:add-cx}
\end{lemma}
\begin{proof}
Suppose $U$ is a local gate on $k$ qubits.
Then $U$ can be expressed as the sum of $4^k$ terms, $U=\sum_{x,y\in \{0,1\}^k}a_{xy}\ket x\bra y\otimes \mathbb I_{n-k}$.
Each of these terms individually can be applied to an \add in polynomial time (\citeauthor{fargier2014knowledge} Table 1 \textbf{CD}), since they are projections, followed by $X$ gates.
Since a constant number of states can be added in polynomial time in \adds (\cref{th:add-add}), the result can be computed in polynomial time.
Since \add supports arbitrary $k$-local gates, in particular it supports all other gates that are mentioned: $H$, $X$, $Y$, $Z$, $T$, $\text{Swap}$ $CZ$.
\end{proof}

\subsection{Easy and hard operations for \qmdd}
\label{app:trac:qmdd}

\begin{lemma}
	\qmdd supports \eq, \samp and \pro.
\label{th:qmdd-samp}
\label{th:qmdd-prob}
\end{lemma}
\begin{proof}
Since \qmdd specializes \limdd (see \cref{thm:limdd-prob}).	
\end{proof}

\begin{lemma}
	\qmdd supports inner product (\inprod) and fidelity (\fid).
	\label{th:qmdd-inprod}
\end{lemma}
\begin{proof}
	We show in \cref{sec:qmdd-to-mps} that a \qmdd can be efficiently and exactly translated to an MPS.
	Since MPS supports inner product and fidelity, the result follows.
\end{proof}

\begin{lemma}
	\textsc{\qmdd} does not support \addi in polynomial time.
\label{thm:qmdd-not-add}
\end{lemma}
\begin{proof}
\citeauthor{fargier2014knowledge} (Thm.~4.9) show that \addi is hard for \qmdd.
\end{proof}

\begin{lemma}
	\textsc{\qmdd} does not support \had in polynomial time and hence neither \loc.
\label{thm:qmdd-not-had}
\label{thm:qmdd-not-loc}
\end{lemma}
\begin{proof}
	By reduction from addition: Take a \qmdd root node $v$ with left child
	$a$ and right child $b$, then
	$\had(v) = H\ket v$ is a new node with a left child $\ket a + \ket b$.
	By choosing $\ket a,\ket b$ to be the states from Fargier et al.'s proof showing that addition is intractable for \qmdds, the state $\ket a+\ket b$ requires an exponential-size \qmdd.
Since \qmdd does not support the Hadamard gate, neither does it support arbitrary local gates (\loc).	
\end{proof}

\begin{lemma}
	\label{thm:qmdd-supports-xyzt}
	\textsc{\qmdd} supports Pauli gates \xyz and \T in polynomial time.
\end{lemma}
\begin{proof}
We will show that we can apply any single-qubit diagonal or anti-diagonal operator $A=\diagg{\alpha}{\beta}$ to an \qmdd in polynomial time.
The result then immediately follows for the special cases of the gates $X,Y,Z$ and $T$.
Applying any diagonal local operator $A = \diagg{\alpha}{\beta}$ to the top qubit is easy:
simply multiply the weights of the low and high edges of the diagram's root node with respectively $\alpha$ and $\beta$. For the anti-diagonal operator $A^T$, we also swap low with high edges.
To apply the local operator on any qubit, simply do the above for all nodes on the corresponding level.

	To see that the resulting \qmdd indeed represents the state $A\cdot\ket\psi$ (or $A^T\ket\psi$), consider the amplitude of any basis state $x\in\{0,1\}^n$.
	The amplitude of $\ket x$ in an \qmdd is the product of the labels found on the edges while traversing the diagram from root to leaf.
	In the new diagram, only the weights have changed, whereas the topology has remained the same.
	If $x_k=0$ (resp. $x_k=1$), then, the $k$-th edge encountered during this traversal is the same in the new diagram as in the old diagram, but the label has been multiplied by $\alpha$.
	Otherwise, if $x_k=1$ (resp. $x_k=0$), then the label is multiplied by $\beta$.
	All the other weights remain the same.
	Therefore, the amplitude of $x$ in the new diagram is equal to the old amplitude multiplied by $\alpha$ (resp. $\beta)$.
\end{proof}

\begin{lemma}
	\qmdd supports controlled-$Z$ in polynomial time.
	\label{thm:qmdd-cz}
\end{lemma}
\begin{proof}
	\cref{alg:qmdd-cz} applies a controlled-$Z$ gate to a \qmdd in time linear in the number of nodes in the \qmdd.
	To show that this is the runtime, we consider the number of times the algorithms \textsc{ApplyControlledZ} and \textsc{ApplyZ} are called.

	For both these algorithms, say that a call is \emph{trivial} if the result is already in the cache, otherwise a call is \emph{non-trivial}.
	Then a trivial call completes in constant time (i.e., in time $\mathcal O(1)$).
	Moreover, the number of trivial calls is at most twice the number of non-trivial calls.
	Therefore, for the purposes of obtaining an asymptotic upper bound on the running time, it suffices to count the number of non-trivial calls to the algorithm.
	
	Thanks to the cache, a given setting of the input parameters $(v,a,b)$ (or $(v,t)$ in the case of \textsc{ApplyZ}) will trigger only one non-trivial call.
	Therefore, the number of non-trivial calls is equal to the number of distinct input parameters.
	But here only $v$ varies, so the number of non-trivial calls is at most the number of nodes in the \qmdd.
	This reasoning holds for both algorithms \textsc{ApplyControlledZ} and \textsc{ApplyZ}.
	Therefore, both subroutines run in time $\mathcal O(m)$, for an \qmdd which contains $m$ nodes.
	
	The correctness of this algorithm follows from the fact that $CZ_{a,b} \ket v = \lambda_0 \ket{v_0} + 
	\lambda_1 Z_b \ket{v_1}$
	where node $v$ is represents an $a-qubit$ state $\lnode[v]{\lambda_0}{v_0}{\lambda_1}{v_1}$ and $Z_b$ means applying the $Z$ gate to the $b$-th qubit. This behavior is implemented by \cref{l:cz}. By linearity, the algorithm is correct for nodes representing $k$-qubit states with $k> a$. This is implemented by \cref{l:ncz}.
\begin{algorithm}
	\caption{Applies a controlled-$Z$ gate to a \qmdd node $v$, with control qubit $a$ and target qubit $b$.
	More specifically, given a \qmdd node $v$, representing the state $\ket{v}$, this subroutine returns a \qmdd edge $e$ representing $\ket{e}=CZ_a^b\ket{v}$.
	We assume wlog that $a>b$, since $CZ_a^b=CZ_b^a$.
	Here $\nindex$ denotes the index of the qubit of $v$, and \textsc{Cache} denotes a hashmap which maps triples to \qmdd nodes.
	The subroutine \textsc{ApplyZ} applies a $Z$ gate to a given target qubit $t$.}
	\label{alg:qmdd-cz}
	\begin{algorithmic}[1]
		\Procedure{ApplyControlledZ}{\qmdd node $v$, qubit indices $a,b$}
			\State Say that node $v$ is $\lnode[v]{\lambda_0}{v_0}{\lambda_1}{v_1}$
			\If{the \textsc{Cache} contains the tuple $(v,a,b)$}
				\Return $\textsc{Cache}[v,a,b]$
			\ElsIf{$\nindex(v)=a$} $r:=\textsc{MakeNode}(\ledge{\lambda_0}{v_0}, \lambda_1\cdot \textsc{ApplyZ}(v_1,b))$\label{l:cz}
			\Else\ $r:=\textsc{MakeNode}(\lambda_0 \cdot \textsc{ApplyCZ}(v_0,a,b), \lambda_1\cdot\textsc{ApplyCZ}(v_1,a,b))$\label{l:ncz}
			\EndIf
			\State $\textsc{Cache}[v,a,b]:=r$
			\State \Return $r$
		\EndProcedure
		\Procedure{ApplyZ}{\qmdd node $v$, qubit index $t$}
			\State Say that node $v$ is $\lnode[v]{\lambda_0}{v_0}{\lambda_1}{v_1}$
			\If{the \textsc{Cache} contains the tuple $(v,t)$}
				\Return $\textsc{Z-Cache}[v,t]$
			\ElsIf{$\nindex(v)=t$} $r:=\textsc{MakeNode}(\ledge {\lambda_0}{v_0}, \ledge {-\lambda_1}{v_1})$
			\Else\ $r:=\textsc{MakeNode}(\lambda_0\cdot \textsc{ApplyZ}(v_0,t), \lambda_1 \cdot \textsc{ApplyZ}(v_1,t))$
			\EndIf
			\State $\textsc{Z-Cache}[v,t]:=r$
			\State \Return $r$
		\EndProcedure
	\end{algorithmic}
\end{algorithm}
\end{proof}

\begin{wrapfigure}[14]{r}{4cm}
\vspace{-1em}
\begin{tikzpicture}[
    scale=0.3,
    every path/.style={>=latex},
    every node/.style={},
    inner sep=0pt,
    minimum size=14pt,
    line width=1pt,
    node distance=.5cm,
    thick,
    font=\footnotesize
    ]

    \node[draw,circle] (a1) {};
    \node[draw,circle, below =of a1] (a2) {};
    \node[draw,circle, below =of a2] (a3) {};
    \node[ below =-.1cm of a2] (a) {$\pmb{\vdots}$};
    \node[draw,circle,rectangle,minimum size=0.4cm, below=of a3] (w1) {$1$};

    \draw[<-] (a1) --++(90:2cm) node[pos=1.4,right] {$+^n$};
    \draw[e0=25] (a1) edge  node[] {} (a2);
    \draw[e1=25] (a1) edge  node[right] {} (a2);
    \draw[e0=25] (a3) edge  node[] {} (w1);
    \draw[e1=25] (a3) edge  node[right] {} (w1);

    \node[draw,circle,right = 1cm of a1] (a1) {};
    \node[draw,circle, below =of a1] (a2) {};
    \node[draw,circle, below =of a2] (a3) {};
    \node[ below =-.1cm of a2] (a) {$\pmb{\vdots}$};
    \node[draw,circle,rectangle,minimum size=0.4cm, below=of a3] (w1) {$1$};

    \draw[<-] (a1) --++(90:2cm) node[pos=1.4,right] {$\Rot^n$};
    \draw[e0=25] (a1) edge  node[] {} (a2);
    \draw[e1=25] (a1) edge  node[right=.3cm] {$e^{i\pi \cdot 2^{-n-1} }$} (a2);
    \draw[e0=25] (a3) edge  node[] {} (w1);
    \draw[e1=25] (a3) edge  node[right=.3cm] {$ e^{i\pi \cdot 2^{-2}}$} (w1);
\end{tikzpicture}
\caption{The states $\ket{+^n}$ and $\ket{\Rot^n}$ state as \qmdd.}
\label{fig:sump}
\end{wrapfigure}
To prove that a single swap operation can explode the \limdd or \qmdd, we first  provide two lemmas.

\begin{lemma}
	\label{thm:sump-small-qmdd}
	For $n\geq 1$, let $\ket{\Rot^n} = \bigotimes_{j=1}^n\left(\ket 0 + e^{i\pi2^{-j-1}}\ket 1\right)$ and $\ket{+^n} = \ket{+}^{\otimes n}$.Then the states $\ket{\Rot^n}$ and $\ket{+^n}$ have a linear-size \qmdd.
\end{lemma}
\begin{proof}
	\cref{fig:sump} provides the \qmdd representing  
	both states.
\end{proof}

\begin{lemma}
	\label{thm:sum-prime-has-exp-limdd}
	The following state has large \limdd, for any variable order in which the qubit in register $B$ comes after the qubits in register $A$.
	\begin{align}
	\ket{\mathit{Sum}'}=\ket{+}_A^{\otimes n}\ket{0}_B + \bigotimes_{j=1}^n(\ket {0}_A+e^{i\pi 2^{-j-1}}\ket {1}_A)\otimes \ket {1}_B
	\end{align}
\end{lemma}
\begin{proof}
	The proof is similar to that of \cref{thm:state-has-exp-limdd} except that we reason about level 2.
\end{proof}

\begin{lemma}
	\label{thm:qmdd-not-swap}
	\qmdd does not support \swap in polynomial time.
\end{lemma}
\begin{proof}
	Let $\ket{+^n}$ and $\ket{\Rot^n}$ be the states from \cref{thm:sump-small-qmdd}, and define the following state $\ket{\rho}$ on $n+2$ qubits,
	\begin{align}
	\ket\rho = \ket0\ket{+^n}\ket 0 + \ket1\ket{\Rot^n} \ket 0
	\end{align}
	Then $\ket\rho$ has a small \qmdd, of only size $\mathcal O(n)$.
	When we swap the first and last qubits, we obtain a state that includes $\ket{\mathit{Sum}'}$ from \cref{thm:sum-prime-has-exp-limdd}:
	\begin{align}
	\text{Swap}_1^n\cdot\ket\rho = \ket 0\otimes \left(\ket{+^n}\ket0 + \ket{\Rot^n}\ket1\right) 
	\end{align}
	The \qmdd of $\text{Swap}_1^n\cdot\ket\rho$ is at least as large as that of \sumstate:
	First, \cite{wegener2000branching} Th.~2.4.1 shows that constraining can never increase the DD size, 
	so we can discard the $\ket 0\otimes $ part (regardless of variable order), as $\text{Swap}_1^n\cdot\ket\rho$ is at least as large as $\ket{\mathit{Sum}'}$.
	Then  \cref{thm:sum-prime-has-exp-limdd} shows that this \limdd has size at least $2^{\Omega(n)}$ for any variable order. Since, \limdd is at least as succinct as \qmdd (see \cref{fig:succinct}), this also holds for \qmdd.
\end{proof}

\begin{lemma}
	\label{thm:qmdd-not-local}
	\qmdd does not support \loc in polynomial time.
\end{lemma}
\begin{proof}
	This is implied by \cref{thm:qmdd-not-swap}.
\end{proof}

\subsection{Easy and hard operations for \limdd}
\label{app:trac:limdd}

\begin{lemma}
	\limdd supports \samp, \pro, \eq and \xyz.
\label{thm:limdd-samp}
\label{thm:limdd-prob}
\label{thm:limdd-eq}
\label{thm:limdd-xyz}
\end{lemma}
\begin{proof}
Vinkhuijzen et al. \cite{vinkhuijzen2021limdd} show that \limdd supports \samp, \pro, \eq and \xyz.	
\end{proof}

Here we show that \limdd also support applying a $T$ gate, but does not support \addi, $H$ and \swap.
In this work, we show that computing the fidelity (and hence the inner product, as we can reduce fidelity to inner product) between two states represented by \limdds is NP-hard.

\begin{lemma}
	\limdd supports Controlled-$Z$ in polynomial time.
	\label{thm:limdd-cz}
\end{lemma}
\begin{proof}
	Vinkhuijzen et al. \cite{vinkhuijzen2021limdd} show how to apply any controlled Pauli gate to a state represented by a \limdd in polynomial time, in the case where the target qubit comes after the control qubit in the variable order of that \limdd.
	However, in the case of the controlled-$Z$, there is no distinction between control and target qubit, since the gate is symmetric.
	Therefore, their analysis applies to all controlled-$Z$ gates.
	In fact, inspecting their method, we see that the \limdd of the resulting state is never larger in size than the \limdd we started with.
\end{proof}

It is known that addition is hard for \qmdd~(see Table 2 in \cite{fargier2014knowledge}). For \limdd, the same was suspected, but not proved in \cite{vinkhuijzen2021limdd}.
We show it here by showing that
$\braket{Z}$-\limdd does not support addition in polytime.
\begin{lemma}
	\limdd does not support \addi in polytime.
\end{lemma}
\begin{proof}
Consider the states $\ket {+^n}$ and $\ket{\Rot^n}$ as defined in \cref{thm:sum-prime-has-exp-limdd}.
Both states have polynomially sized \qmdds as shown in \cref{thm:sump-small-qmdd}.
Since \limdd is at least as succinct as \qmdd (see \cref{fig:succinct}), the \limdds representing these states are also small.
However, their sum is the state $\ket{Sum}=\ket {+^n} + \ket{\Rot^n}$,  which has an exponential-size
\limdd relative to every variable order by \cref{thm:state-has-exp-limdd}.
\end{proof}

\begin{lemma}
	\limdd does not support \had in polynomial time, and hence neither does it support \loc.
\end{lemma}
\begin{proof}
	Since Hadamard can be used together with measurement (called \emph{conditioning} by Fargier \cite{fargier2014knowledge}) to realize state addition as
	explained in \cref{sec:tractable-queries}, it is also intractable. (Recall also from \cite{wegener2000branching} Th.~2.4.1 that conditioning never increases DD size; this is true in particular for \limdds)
\end{proof}

We now prove that the fidelity of \limdds cannot be computed in polynomial time, under common assumptions of complexity theory.

\paragraph{\limdd FIDELITY is hard to compute.}
We show that \limdd FIDELITY cannot be computed in polynomial time, unless the Exponential Time Hypothesis (ETH) is false.
This proof implies that inner product is hard, since fidelity reduces to inner product.
Proving hardness of inner product is also a specialized case of the below construction, which does not require our newly defined EOSD problem (see below) but only the well-known hard problem of counting even subgraphs of a certain size (\#EVEN SUBGRAPHS).

The proof of \limdd FIDELITY hardness proceeds in several steps.
The starting point is Jerrum and Meeks' result that the problem \#EVEN SUBGRAPHS cannot be solved in polynomial time unless ETH is false (\cref{thm:even-subgraphs-ETH-hard}).
We introduce a problem we call EVEN ODD SUBGRAPHS DIFFERENCE (EOSD).
We give a reduction from \#EVEN SUBGRAPHS to EOSD, thus showing that EOSD cannot be solved in polynomial time, under the same assumptions (\cref{thm:eosd-ETH-hard}).
This step is the most technical part of the proof.
Finally, we give a reduction form EOSD to \limdd FIDELITY, thus obtaining the desired result, that \limdd FIDELITY cannot be computed in polynomial time (to a certain precision), unless ETH is false (\cref{thm:limdd-fidelity-ETH-hard}).
In this step, we use the fact that \limdds can efficiently represent Dicke states and graph states (a type of stabilizer state).
Specifically, we will show that computing the fidelity between these states essentially amounts to solving EOSD for the given graph state.
Dicke states were first studied by \citeauthor{dicke1954coherence}; see also \cite{bartschi2019deterministic}.

We first formally define the three problems above, including computing the fidelity of two \limdds.
We will need the following terminology for graphs.
For an undirected graph $G=(V,E)$ and a set of vertices $S\subseteq V$, we denote by $G[S]$ the subgraph induced by $S$.
If $|S|=k$, then we say that $G[S]$ is a  $k$-induced subgraph, and we say that it is an \emph{even} (resp. \emph{odd}) subgraph if $G[S]$ has an even (resp. odd) number of edges.
We let $e(G,k)$ (resp. $o(G,k)$) denote the number of even (resp. odd) $k$-induced subgraphs of $G$.

\begin{description}
	\item[\limdd FIDELITY.]~\\
	\textbf{Input:} Two \limdds, representing the states $\ket{\phi},\ket\psi$\\
	\textbf{Output:} The value $|\braket{\phi|\psi}|^2$ to $2n$ bits of precision.

	\item[\#EVEN SUBGRAPHS]~\\
	\textbf{Input:} A graph $G=(V,E)$, an an integer $k$\\
	\textbf{Output:} The value $e(G,k)$.
	
	\item[EVEN ODD SUBGRAPH DIFFERENCE (EOSD).]~\\
	\textbf{Input:} A graph $G=(V,E)$, and an integer $k$.\\
	\textbf{Output:} The value $|e(G,k)-o(G,k)|$, i.e., the absolute value of the difference between the number of even and odd induced $k$-subgraphs of $G$.
\end{description}

\begin{lemma}[\citeauthor{jerrum2017parameterised}]
	\label{thm:even-subgraphs-ETH-hard}
	If \#EVEN SUBGRAPHS is polytime, then ETH is false.
\end{lemma}
\begin{proof}
	\citeauthor{jerrum2017parameterised} showed that counting the number of even induced subgraphs with $k$ vertices is $\#\textsc{W}[1]$-hard.
	Consequently, there is no algorithm running in time $poly(n)$ (independent of $k$) unless the exponential time hypothesis fails.
\end{proof}

\begin{lemma}
	\label{thm:limdd-fidelity-ETH-hard}
	There is no polynomial-time algorithm for \limdd FIDELITY, i.e., for computing fidelity between two \limdds to $2n$ bits of precision, unless the Exponential Time Hypothesis (ETH) fails.
\end{lemma}
\begin{proof}
	Suppose there was such a polynomial-time algorithm, running in time $\mathcal O(n^c)$ for some constant $c\geq 1$.
	We will show that then EOSD can be solved in time $\mathcal O(n^{c})$ (independent of $k$), by giving a reduction from EOSD to \limdd FIDELITY.
	From \cref{thm:eosd-ETH-hard}, it would then follow that ETH is false.

	The reduction from EOSD to \limdd FIDELITY is as follows.
	
	Let $G$ be an input graph on $n$ vertices $V$ and $0\leq k\leq n$ an integer.
	Let $\ket G$ be the graph state corresponding to $G$ \cite{nest2004graphical}, so that
	\begin{align}
	\ket G = \frac1{2^{n/2}}\sum_{S\subseteq V}(-1)^{|G[S]|}\ket S
	\end{align}
where $|G[S]|$ denotes the number of edges in the $S$-induced subgraph of $G$, and $\ket S$ denotes the computational-basis state $\ket{S} = \ket{x_1} \otimes \ket{x_2} \otimes \dots \otimes \ket{x_n}$ with $x_j = 1$ if $j \in S$ and $x_j = 0$ otherwise.
	Let $\ket {D_n^k}$ be the Dicke state \cite{dicke1954coherence}.
	\begin{align}
	\ket{D_n^k} = \frac{1}{\sqrt{\binom nk}}\sum_{\vec x\in \{0,1\}^n \text{~with~} |\vec x|=k}\ket{\vec x}
	\end{align}
	Both these states have small \limdds: 
	\begin{description}
	\item[Dicke state] \citeauthor{bryant86} gives a construction for BDDs to represent the function $f_k\colon \{0,1\}^n\to \{0,1\}$, with $f_k(x)=1$ iff $|x|=k$.
	This is precisely the amplitude function of the Dicke state $\ket{D_n^k}$ (up to a factor $1/\sqrt{\binom nk}$).
	This construction also works for \limdds, by simply setting all the edge labels to the identity, and using root label $1/\sqrt{\binom nk}\cdot \mathbb I^{\otimes n}$.
	\item[Graph state] \citeauthor{vinkhuijzen2021limdd} show how to efficiently construct a \limdd for any graph state.
	\end{description}
It is straightforward to verify that the fidelity between $\ket{D_n^k}$ and $\ket{G}$ is related to the subgraphs of $G$, as follows,
	\begin{align}
	\braket{D_n^k|G} = \frac{1}{\sqrt{\binom nk2^n}} \sum_{S\subseteq V:|S|=k}(-1)^{|G[S]|} = \frac{1}{\sqrt{\binom nk2^n}} \left(e(G,k)-o(G,k)\right)
	\end{align}
Hence,
	\begin{align}
	\underbraceset{\textnormal{solution to EOSD}}{|e(G,k)-o(G,k)|}=\sqrt{\binom nk2^n \underbraceset{\textnormal{Fidelity}}{|\braket{D_n^k|G}|^2}}
	\end{align}
Since $|\braket{D_n^k|G}|^2$ denotes the fidelity between $\ket{D_n^k}$ and $\ket{G}$, and $|e(G,k)-o(G,k)|$ denotes the quantity asked for by the EOSD problem, this completes the reduction.
The overhead of constructing the \limdds from the description of the Dicke and graph states takes linear time in the size of the resulting \limdd. So, if the fidelity of two \limdds is computed in polynomial time, say, in time $\mathcal O(n^c)$, then also the quantity $|e(G,k)-o(G,k)|$ is computed in time $\mathcal O(n^c)$; thus, EOSD is solved in time $\mathcal O(n^c)$.
Lastly, we address the number of bits of precision required.
In order to exactly compute the integer $|e(G,k)-o(G,k)|$, it is necessary to compute the fidelity $|\braket{D_n^k|G}|^2$ with a precision of at least one part in $\binom nk2^n$.
Put another way, the required number of bits of precision is $\log_2(\binom nk\cdot 2^n)\leq \log_2(2^n\cdot 2^n)=2n$.
Summarizing, computing the fidelity of (the states represented by) two \limdds representing a graph state and a Dicke state, to $2n$ bits of precision, is not possible in polynomial time, unless ETH fails.
\end{proof}

\begin{lemma}
	\label{thm:eosd-ETH-hard}
	There is no polynomial-time algorithm for EOSD, unless ETH is false.
\end{lemma}
\begin{proof}
	We provide an efficient reduction (in \cref{alg:compute-even-subgraphs}) from \#EVEN SUBGRAPHS: the problem, on input an undirected graph $G$ and a parameter $k\in \{0,1, 2, \dots, |V|\}$, of computing the number of $k$-vertex induced subgraphs which have an even number of edges.
	It follows that, if EOSD can be computed in polynomial time, then \cref{alg:compute-even-subgraphs}, which computes \#EVEN SUBGRAPHS, also runs in polynomial time.
	Jerrum and Meeks \cite{jerrum2017parameterised} show that \#EVEN-SUBGRAPHS cannot be computed in polynomial time unless ETH is false (\cref{thm:even-subgraphs-ETH-hard}).
	Therefore, if EOSD could be computed in polynomial time, then ETH would be false.

	The algorithm \textsc{CountEvenSubgraphs} (\cref{alg:compute-even-subgraphs}) takes as parameters a graph $G$ and an integer $k\geq 0$, and outputs $e(G,k)$, the number of even $k$-induced subgraphs of $G$, thus solving \#EVEN SUBGRAPHS.
	This algorithm uses at most $2n$ invocations of a subroutine \textsc{EvenOddSubgraphsDifference}; therefore, if the subroutine \textsc{EvenOddSubgraphsDifference} runs in polynomial time, then so does \textsc{CountEvenSubgraphs}.
	
	Let us briefly sketch the idea behind the algorithm, before we give a formal proof of correctness.
	First, we know that $e(G,k) + o(G,k)=\binom nk$, since each subgraph is either even or odd, and $G$ has $\binom nk$ different $k$-induced subgraphs in total.
	Thus, if we knew the (possibly negative) difference $\zeta_k=e(G,k) - o(G,k)$, then we know the sum and difference of $e(G,k)$ and $o(G,k)$, so we could compute the desired value $e(G,k)=\frac 12(\binom nk+\zeta_k)$.
	Unfortunately, \textsc{EvenOddSubgraphsDifference} only tells us the absolute value, $|\zeta_k|$.
	Fortunately, we know that $e(G,0)=1$ and $o(G,0)=0$, so $\zeta_0=1-0=1$ (namely, there is only one induced subgraph with $0$ vertices, and it has $0$ edges, which is even).
	We now bootstrap our way up, computing $\zeta_j$ for $j=1,\ldots, k$ using the previously known results.
	The key ingredient is that, by adding isolated vertices to the graph and querying \textsc{EvenOddSubgraphsDifference} on this new graph, we can discover the the absolute difference $|\zeta_j+\zeta_{j-1}|$, which allows us to compute the values $\zeta_j$.
	
	\begin{algorithm}[b!]
		\caption{An algorithm which computes the number of even $k$-induced subgraphs using at most $2n$ calls to a subroutine \textsc{EvenOddSubgraphsDifference}, which returns $|e(G,k)-o(G,k)|$ on input $(G,k)$.}
		\label{alg:compute-even-subgraphs}
		\begin{algorithmic}[1]
			\Procedure{CountEvenSubgraphs}{$G=(V,E),k$}
				\Statex Output: The number of even induced subgraphs of $G$ with $k$ vertices
				\State $d_0:=1$ \Comment{$d$ is an array of $k+1$ integers}
				\label{algline:delta-0-is-1}
				\State $\ell:=0$   \Comment{Last iteration when $\zeta_j = 1$}
				\For{$j:=1,\ldots, k$}
					\State $q:=\textsc{EvenOddSubgraphsDifference}(G, j)$
					\label{algline:compute-q}
					\If{$q=0$}\Comment{There are equally many even as odd subgraphs}
						\State $d_j:=0$ 
						\label{algline:set-delta-ell-zero}
					\Else\Comment{Else we have to figure out whether there more even or odd subgraphs:}
						\State $G':=(V\cup \{v_1',\ldots,v_{j-\ell}'\}, E)$ \Comment{Add $j-\ell$ new isolated vertices}
						\State $p:=\textsc{EvenOddSubgraphsDifference}(G', j)$
						\label{algline:compute-p}
						\State $d_j:=\begin{cases}
						q & \text{if }|d_\ell+q|=p \\
						-q& \text{if }|d_\ell-q|=p
						\end{cases}$
						\label{algline:compute-delta-ell-nonzero}
						\State $\ell:=j$		 \Comment{Since iteration $j$ is the latest iteration having $\zeta_j = 1$}
					\EndIf
				\EndFor
				\vspace{0.5em}
				\State \Return $\frac 12\left(\binom nk+d_k\right)$
			\EndProcedure
		\end{algorithmic}
	\end{algorithm}
	\textbf{Correctness of the algorithm.}
	We now prove that the algorithm \textsc{CountEvenSubgraphs} outputs the correct value.
	Let $(G,k)$ be the input to the algorithm.
	For $j=0,\ldots, k$, let $\zeta_j = e(G,j)-o(G,j)$.
	We will show that, for each $j=1,\ldots, k$, the algorithm sets the variable $d_j$ to the value $\zeta_j$ in the $j$-th iteration of the for-loop.
	The proof is by induction on $j$.
	In the induction hypothesis, we include also that the variable $\ell$ is always the largest value below $j$ satisfying $\zeta_\ell\ne 0$  as in \cref{eq:ih2} (this value is well-defined, since $1=\zeta_0\ne 0$, so we have $0\leq \ell<j$).
	\begin{align}\label{eq:ih2}
	\ell=\max\{0\leq \ell<j \ |\  \zeta_\ell\ne 0\}
	\end{align}
	For the base case, where $j=0$, it suffices to note that there is only one set with zero vertices -- the empty set -- which induces the empty graph, which contains an even number of edges.
	Therefore, $\zeta_0=1$, which the algorithm sets on \cref{algline:delta-0-is-1}.
	Finally, $\ell$ is correctly set to 0.

	For the induction case $j\geq 1$,  the variables $d_t$ have been set to $d_t=\zeta_t$ for $t=0,\ldots, j-1$ and $\ell$ satisfies \cref{eq:ih2} from the induction hypothesis.
	Consequently, we have $\zeta_{\ell+1}=\cdots=\zeta_{j-1}=0$.
	If $\zeta_j=0$, then the algorithm sets $q:=|\zeta_j|=|0|=0$ on \cref{algline:compute-q}, so the algorithm sets $d_j$ correctly on \cref{algline:set-delta-ell-zero}, and correctly leaves $\ell$ untouched ($\ell$ remains unchanged from the $j-1$-th to the $j$-th iteration).
	Otherwise, if $\zeta_j\ne 0$, the algorithm adds $j-\ell$ new, isolated vertices to $G$, obtaining the new graph $G'=(V_G\cup \{v_1',\ldots, v_{j-\ell}'\},E_G)$.
	On \cref{algline:compute-p}, it computes the value $|e(G',j)-o(G',j)|$ of this graph.
	Since this expression also sums over induced subgraphs of $G'$ that contain isolated vertices, this value can be expressed as follows:
	\begin{align}
	p := & |e(G',j)-o(G',j)| \\
	= & \left| \sum_{a=\ell}^j \binom{j-\ell}{j-a}(e(G,a)-o(G,a)) \right| = \left| \sum_{a=\ell}^j \binom{j-\ell}{j-a}\zeta_a \right| \label{eq:eosd1}
	\\
	= & \left| \binom{j-\ell}{j-\ell}\zeta_\ell + \binom{j-\ell}{0}\zeta_j\right|  =  \left|\zeta_\ell+\zeta_j\right|  \label{eq:eosd2}
	\end{align}
	We noted that $\zeta_{\ell+1}=\cdots=\zeta_{j-1}=0$; therefore, these terms vanish from the summation (step from \cref{eq:eosd1} to \cref{eq:eosd2}), so that only $p=|\zeta_j+\zeta_\ell|$ remains.
	Since we now know the values of $\zeta_\ell,|\zeta_j|$ and $|\zeta_j+\zeta_\ell|$ and since $\zeta_\ell\ne 0$, we can infer the value of $\zeta_j$, which is done on \cref{algline:compute-delta-ell-nonzero}.
	We conclude that each variable $d_j$ is correctly set to $\zeta_j$, concluding the proof by induction. Also, since $\zeta_j \neq 0$, $\ell$ is correctly set to $j$.
	
	Lastly, we show that the value returned by the algorithm is indeed the number $e(G,k)$.
	Suppose that $d_k=\zeta_k=e(G,k)-o(G,k)$.
	We know that $e(G,k)+o(G,k)=\binom nk$.
	That is, we know both the sum and the difference of $e(G,k),o(G,k)$; therefore we can compute them both.
	By adding these two equations and solving for $e(G,k)$, we obtain $e(G,k)=\frac 12\left(\binom nk+d_k\right)$, which is the value returned by the algorithm.
\end{proof}

\begin{lemma}
	\label{thm:limdd-t-gate}
	\limdd supports the $T$-gate.
\end{lemma}
\begin{proof}
	\cref{alg:limdd-t-gate} applies an arbitrary diagonal gate $D=\begin{smallmat}\rho & 0 \\ 0 & \omega\end{smallmat}$ to a state represented by a \limdd.
	To apply a $T$-gate, one calls the algorithm with $D=T=\diag{e^{i\pi/4}}$.
	We now show that the algorithm runs in polynomial time.
	First, since each recursive call takes $\mathcal O(1)$ time, for the purposes of estimating runtime it suffices to count the number of recursive calls.
	The cache stores all tuples of nodes and matrices with which the algorithm is called, so for the purposes of estimating the runtime it suffices to count the number of \emph{distinct} recursive calls.
	To this end, we note that the recursive calls to the algorithm only receive two different matrices, namely $\begin{smallmat}\rho & 0 \\ 0 & \omega\end{smallmat}$ and $\begin{smallmat}\omega & 0 \\ 0 & \rho\end{smallmat}$.
	The nodes that are passed as argument $v$ are nodes that are already in the diagram.
	Therefore, if the diagram contains $m$ nodes, then at most $2m$ distinct recursive calls are made.
	We conclude that the runtime is polynomial (indeed, linear), in the size of the diagram.
	\begin{algorithm}
	\caption{Applies a diagonal gate $D$ to qubit $k$ of a state represented by a \limdd.}
	\label{alg:limdd-t-gate}
	\begin{algorithmic}[1]
		\Procedure{ApplyDiagGate}{\limdd Node $v=\lnode{\lambda_0A_0}{v_0}{\lambda_1 A_1}{v_1}$, gate $D$, qubit $k$}
			\Statex with $D=\begin{smallmat}\rho & 0 \\ 0 & \omega\end{smallmat}$
			\Statex Node $v$ represents a state on $n$ qubits
			\If{\textsc{Cache} contains the tuple $(v,D)$}
				\Return $\textsc{Cache}[v,D]$
			\ElsIf{$k=n$}
				\State \Return $\lnode{\rho\lambda_0A_0}{v_0}{\omega\lambda_1A_1}{v_1}$
			\Else
				\State gate $E:=\begin{smallmat}\omega & 0 \\ 0& \rho\end{smallmat}$
				\For{$i=0,1$}
				\State gate $F_i :=\begin{cases}
				D     & \text{if }A^k_i\in \{I,Z\} \\
				E     & \text{if }A^k_i\in \{X,Y\}
				\end{cases}$\Comment{Here $A^k_i$ denotes the $k$-th qubit of the Pauli operator $A^i$}
				\State Node $u_i:=\textsc{ApplyTGateTo\limdd}(v_0,F_i,k)$
				\EndFor
				\State Node $r:=\lnode{\lambda_0A_0}{u_0}{\lambda_1A_1}{u_1}$
				
				\State $\textsc{Cache}[v,D]:=r$
				\State \Return $r$
			\EndIf
		\EndProcedure
	\end{algorithmic}
	\end{algorithm}
\end{proof}

\subsection{Supported operations of MPS}
\label{ops:mps}
\label{app:trac:mps}

\citeauthor{vidal2003efficient} shows that MPS supports efficient application of a single one-qubit gate or two-qubit gate on consecutive qubits, which includes \xyz, \had, \T.
This extends to any two-qubit operation on any qubits \cite{perez2007matrix}, particularly including \swap and \cz gates.
These algorithms are extendable to $k$-local gates on adjacent qubits, which does not increase the largest matrix dimension $D$ to more than $D^k$.
The algorithm consists of merging the $k$ tensors (the $j$-th tensor combines the two matrices $A^0_j$ and $A^1_j$) into a single large one, applying the gate to the large tensor, followed by splitting the tensor again into $k$ matrices $A_j^0$ and $A_j^1$ again by use of the singular-value decomposition (for details on the merging and splitting see e.g. \cite{dang2019optimising}).
The largest matrix dimension during this process does not increase above $D^k$.
Using $\swap$ gates, one thus implements $\loc$ on any qubits.
We give a direct proof of the support for addition below (\cref{thm:mps-addition}).

\citeauthor{orus2014practical} gives an accessible exposition of TN, of which MPS is a special case. He explains how to compute the inner product in polynomial time. 
Thus, MPS also supports \pro.
\samp can be done by a Markov Chain Monte Carlo approach, invoking \pro as subroutine.
Since inner product is supported, so is \eq: MPS $M$ and $M'$ are equivalent iff $\frac{|\langle M | M' \rangle|^2}{\langle M | M \rangle \cdot \langle M' | M' \rangle} = 1$.

\begin{lemma}
	\label{thm:mps-addition}
	MPS supports addition in polynomial time.
\end{lemma}
\begin{proof}
	Let $A,B$ be MPSs.
	Then a new MPS $C$ representing $\ket C=\ket A+\ket B$ can be efficiently constructed as follows, for $x=0,1$ and $j=2,\ldots, n-1$:
	\begin{align}
	C_n^x = \begin{smallmat}
	A_n^x & B_n^x
	\end{smallmat}
	& &
	C_j^x = \begin{smallmat}
	A_j^x & 0 \\ 0 & B_j^x
	\end{smallmat} 
	&& 
	C_1^x=\begin{smallmat}
	A_1^x  \\ B_1^x
	\end{smallmat} 
	\end{align}
\end{proof}

\subsection{Supported operations of RBM}
\label{app:trac:rbm}

\citeauthor{jonsson2018neural} show that RBM supports Pauli gates, the controlled-$Z$ gate and the $T$-gate (and, in fact, arbitrary phase gates).
There is at the moment no efficient exact algorithm for the \had gate, which would make the list of supported gates universal.
Hence there is at the moment no exact efficient algorithm for \loc either.
\samp is supported for any $n$-qubit RBM $M$, see e.g. Appendix B of \cite{jonsson2018neural} and references therein, by performing a Markov Chain Monte Carlo algorithm (e.g. Metropolis algorithm) where the Markov Chain state space consists of all bit strings $x \in \{0, 1\}^n$, and the corresponding unnormalized probability $|\langle x |M\rangle |^2$ of each state is efficiently computed using \cref{eq:rbm-coefficient}.
No exact algorithm for \eq is known (in fact, the related problem of identity testing when one only has sampling access to one of the two RBMs is already computationally hard~\cite{blanca2021hardness}).
Although no exact algorithm for \inprod is known, it can be approximated using \samp as subroutine (see e.g. \cite{wu2020artificial}).
Furthermore, \pro can be approximated by computing the normalization factor $1 / \langle M | M \rangle$ using the (exact or approximate) algorithm for \inprod, while the relative outcome probabilities are defined in \cref{eq:rbm-coefficient}.

\begin{lemma}
	\label{thm:rbm-supports-swap}
	RBM supports \swap.
\end{lemma}
\begin{proof}
	In order to effect a swap between qubits $q_1$ and $q_2$, we simply exchange rows $q_1$ and $q_2$ in the matrix $W$ and the vector $\vec{\alpha}$, obtaining $W^\prime$ and $\vec{\alpha}^\prime$.
	Then $\mathcal M^\prime=(\vec{\alpha}^\prime,\vec\beta,W^\prime,m)$ has $\ket{\mathcal M^\prime}=\textsc{Swap}(q_1,q_2)\cdot\ket{\mathcal M}$.
\end{proof}

\citeauthor{torlai2018neural} note that RBMs can exactly represent Dicke states.
In \cref{thm:rbm-represents-dicke}, we give another construction of succinct RBMs for Dicke states, where the number of hidden nodes grows linearly with the number of visible nodes.
\begin{lemma}
	\label{thm:rbm-represents-dicke}
	An RBM can exactly represent any Dicke state, using only $2n$ hidden nodes.
\end{lemma}
\begin{proof}
	We will construct an RBM with $2n$ hidden nodes representing $\ket{D_n^k}$.
	
	For each $j\in \{0,1,\ldots, n\}\setminus\{k\}$, our construction will use two hidden nodes.
	Fix such a $j$.
	Then the first hidden node is connected to each visible node with weight $i\pi / n$, and has bias $b_j=i\pi (1-j/n)$.
	The second hidden node is connected to each visible node with weight $-i\pi/n$ and has bias $b_j=-i\pi (1-j/n)$.
	Since the weights on all edges incident to a given hidden node are the same, the term it contributes depends only on the weight of the input (i.e., the number of zeroes and ones).
	Thus, these two nodes contribute a multiplicative factor $(1+e^{i\pi (1+|x|/n+j/n)})$ and $(1+e^{-i\pi(1+|x|/n-j/n)})$, respectively.
	Multiplying these together, the two terms collectively contribute a multiplicative term of $2+2\cos(\pi(1+|x|/n-j/n)=2-2\cos(\pi(|x|-j)/n)$, which is $0$ iff $|x|=j$ and nonzero otherwise.
	Let $a$ be the constant $a=\prod_{j=0,j\ne k}^n 2-2\cos(\pi(k-j)/n)$, i.e., the product of all terms when $|x|=k$.
	Then the RBM represents the following state unnormalized function $\Psi\colon \{0,1\}^n\to\mathbb C$, which we then normalize to obtain the state $\ket{\Psi}$:
	\begin{align}
	\Psi(x) = \begin{cases}
	a & \text{if }|x|=k \\
	0 & \text{otherwise}
	\end{cases} & &
	\ket{\Psi} = \frac{1}{a\sqrt{\binom nk}} \sum_x \Psi(x)\ket x
	\end{align}
	The normalized state $\ket{\Psi}$ represents exactly the Dicke state $\ket{D_n^k}$, since its amplitudes are equal to $1/\sqrt{\binom nk}$ when $|x|=k$ and zero otherwise.
\end{proof}

Since RBM can succinctly represent both Dicke states (\cref{thm:rbm-represents-dicke}) and graph states, a subset of stabilizer states~\cite{zhang2018efficient}, the proof for hardness of \limdd FIDELITY (\cref{thm:limdd-fidelity-ETH-hard}) is also applicable to RBM.

\begin{corollary}
	\label{thm:rbm-fidelity-ETH-hard}
	There is no polynomial-time algorithm for RBM FIDELITY, i.e., for computing fidelity between two RBM to $2n$ bits of precision, unless the Exponential Time Hypothesis (ETH) fails.
\end{corollary}

}
{\section{Proofs of \cref{sec:rapidity}}
\label{sec:proofs-of-sec-rapidity}

\cref{sec:proof-rapidity-preorder} proves that our rapidity definition is a preorder and is equivalent to the one given by \cite{lai2017new} for canonical data structures.

\cref{sec:proof-sufficient-condition} provides the proof for the sufficient condition for rapidity and \cref{sec:proofs-rapidity-relations}-\ref{sec:transformations-between-QDDs} prove its application to the data structures studied in this paper.

\subsection{Rapidity Definition}
\label{sec:proof-rapidity-preorder}

We now show that rapidity is a preorder over data structures and that the definition of \cite{lai2017new} can be considered a special case for canonical data structures.
For convenience, we restate the definition of rapidity.

\defrapidity*

\rapidityispreorder*
\begin{proof}
	We first show that rapidity is reflexive, next we show that it is transitive.
	
	\textbf{Rapidity is reflexive.}
	It suffices to show that rapidity is a reflexive relation on algorithms performing a given operation.
	Let $D$ be a data structure, $OP$ an operation and $ALG$ an algorithm performing $OP(D)$.
	Then $ALG$ is at most as rapid as itself if there exists a polynomial $p$ such that for each input $\phi$ there exists an equivalent input $\psi$ with $time(ALG,\phi)\leq p(ALG,\psi)$.
	We may choose the polynomial $p(x)=x$, and we may choose $\psi:=\phi$.
	Then the statement reduces to the trivial statement $time(ALG,\phi)=time(ALG,\psi)=p(ALG,\psi)$.
	
	\textbf{Rapidity is transitive.}
	It suffices to show that rapidity is a transitive relation on algorithms.
	To this end, let $D_1,D_2,D_3$ be data structures, $OP$ an operation and $ALG_1,ALG_2,ALG_3$ algorithms performing $OP(D_1),OP(D_2),OP(D_3)$, respectively.
	Suppose that $ALG_1$ is at most as rapid as $ALG_2$ and $ALG_2$ is at most as rapid as $ALG_3$.
	We will show that $ALG_1$ is at most as rapid as $ALG_3$.
	By the assumptions above, there are polynomials $p$ and $q$ such that (i) for each input $\phi$ there exists an equivalent input $\psi$ such that $time(ALG_2,\psi)\leq p(time(ALG_1,\phi))$; and (ii) for each input $\psi$ there exists an equivalent input $\gamma$ such that $time(ALG_3,\gamma)\leq q(time(ALG_2,\psi))$.
	
	Put together, for every input $\phi$ there exist equivalent inputs $\psi$ and $\gamma$ such that $time(ALG_3,\gamma)\leq q(time(ALG_2,\psi))\leq q(p(time(ALG_1,\phi)))$.
	Letting the polynomial $\ell(x)=q(p(x))$, we obtain that for every $\phi$ there exists an equivalent $\gamma$ such that $time(ALG_3,\gamma)\leq \ell(ALG_1,\phi)$.
\end{proof}

We note that an alternative definition of rapidity~\cite{lai2017new}, which always allows $ALG_2$ to read its input by requiring $\text{time}(ALG_2,y)\leq p(\text{time}(ALG_1,x)+|y|)$ instead of $\text{time}(ALG_2,y) \leq p(\text{time}(ALG_1,x))$, is not transitive for query operations:

Consider the data structure Padded \qmdd, (\textsf P\qmdd) which is just a QMDD, except that a string of $2^{2^n}$ "0"'s have been concatenated to the end of the \qmdd representation, where $n$ is the number of qubits. 

Under the alternative rapidity relation $ \geq_r^{\text{alt}}$, both \add and \qmdd are at least as rapid as \textsf P\qmdd, because the \add algorithm is allowed to run for $poly(2^{2^n})$ time. But \textsf P\qmdd is also at least as rapid as \qmdd, because algorithms for \textsf P\qmdd don't need to read the whole $2^{2^n}$-length input --- they only read the \qmdd at the beginning of the string. Put together, this leads to:

$$ \text{\add} \geq_r^{\text{alt}} \text{\textsf P\qmdd} \geq_r^{\text{alt}} \text{\qmdd}
\text{~~~~~and~~~~~}
\text{\add} \not\geq_r^{\text{alt}} \text{\qmdd}.$$

Next, we show that our definition of rapidity is equivalent to Lai et al.'s definition of rapidity in the case when both data structures are canonical and we restrict our attention to only those algorithms which run in time at least $m$ where $m$ is the size of the input.
For convenience, we restate Lai et al.'s definition here.
\begin{definition}[Rapidity for canonical data structures \cite{lai2017new}]
	A $c$-ary operation $OP$ on a canonical language $L_1$ is \emph{at most as rapid as} $OP$ on another canonical language $L_2$, iff for each algorithm $ALG$ performing $OP$ on $L_1$ there exists some polynomial $p$ and some algorithm $ALG_2$ performing $OP$ on $L_2$ such that for every valid input $(\phi_1,\ldots,\phi_c, \alpha)$ of $OP$ on $L_1$ and every valid input $(\psi_1,\ldots, \psi_c,\alpha)$ of $OP$ on $L_2$ satisfying $\phi_i\equiv \psi_i$ ($1\leq i\leq c$), $ALG_2(\psi_1,\ldots, \psi_c,\alpha)$ can be done in time $p(t+|\phi_1|+\cdots+|\phi_c|+|\alpha|)$, where $\alpha$ is any element of supplementary information and $t$ is the running time of $ALG(\phi_1,\ldots,\phi_c,\alpha)$.
\end{definition}
Lai et al. use several minor differences in notation.
First, they speak of \emph{valid} inputs (because they consider data structures which cannot represent all objects), whereas we do not; they use an element of supplementary information $\alpha$ as part of the input, whereas we omit such an element; they write $\phi_i\equiv \psi_i$ where we write $\ket{\phi_i}=\ket{\psi_i}$; lastly they speak of a \emph{language} whereas we speak of a \emph{data structure}.
Since these differences between the notation are inconsequential, it will be convenient to rephrase the definition of Lai et al. using the notation of this paper, as follows:
\begin{definition}[Rapidity of canonical data structures, rephrased]
	\label{def:rapidity-lai-rephrased}
	In the following, $ALG_1$, $ALG_2$ are algorithms which perform $OP$ on canonical data structures $D_1,D_2$, respectively.
	
	\begin{enumerate}[(a)]
	\item An algorithm $ALG_1$ is \emph{at most as rapid as} an algorithm $ALG_2$ iff there is a polynomial $p$ such that for each input $\phi$ and for each equivalent input $\psi$, it holds that $time(ALG_2,\psi)\leq p(time(ALG_1,\phi)+|\phi|)$.
	\label{item:rapidity-lai-algorithm}
	\item A canonical data structure $D_1$ is \emph{at most as rapid as} a canonical data structure $D_2$ for an operation $OP$ if for each algorithm $ALG_1$ performing $OP$ on $D_1$ there is an algorithm $ALG_2$ performing $OP$ on $D_2$ such that $ALG_1$ is at most as rapid as $ALG_2$.
	\label{item:rapidity-lai-ds}
	\end{enumerate}
\end{definition}

\begin{lemma}
	\label{thm:equivalence-of-rapidity-definitions}
	\cref{def:rapidity} is equivalent to the definition of \cite{lai2017new} (\cref{def:rapidity-lai-rephrased}) in the case when two data structures $D_1,D_2$ are both canonical and where we restrict our attention to algorithms  whose runtime is at least $m$, where $m$ is the size of the input.
\end{lemma}
\begin{proof}
	Let $D_1,D_2$ be two canonical data structures.
	We will show that $D_1$ is at most as rapid as $D_2$ according to \cref{def:rapidity} if and only if the same is true according tot \cref{def:rapidity-lai-rephrased}.
	Since items \ref{def:rapidity}.\ref{item:rapidity-operations} and \ref{def:rapidity-lai-rephrased}.\ref{item:rapidity-lai-ds} are equivalent, it suffices to show that the two definitions are equivalent for \emph{algorithms} rather than \emph{data structures}.
	That is, we will show that an algorithm $ALG_1$ is at most as rapid as $ALG_2$ according to \cref{def:rapidity} if and only if the same is true according to \cref{def:rapidity-lai-rephrased}.
	
	Abusing notation, we write 
	$|(\phi_1,\ldots, \phi_c)|$ instead of $|\phi_1| + \ldots + |\phi_c|$, etc.
	In this proof, we will assume without loss of generality that all polynomials $p$ are monotonically increasing (i.e., $p(x)\leq p(y)$ if $x\leq y$).
	Namely, if $p$ is a polynomial which does not monotonically increase, then use instead the polynomial $p'(x)=p(x)+x^k$ for sufficiently large $k$.
	
	\textbf{Direction \emph{if}.}
	Let $ALG_1,ALG_2$ be algorithms performing $OP$ on canonical data structures $D_1,D_2$, respectively, such that $ALG_1$ is at most as rapid as $ALG_2$ according to \cref{def:rapidity-lai-rephrased}.
	Then there is a polynomial $p$ such that $time(ALG_2,\psi)\leq p(time(ALG_1,\phi)+|\phi|)$ for all equivalent inputs $\phi,\psi$.
	Since the data structures $D_1,D_2$ can represent all quantum state vectors, there certainly \emph{exists} an equivalent $\psi$ to any $\phi$; indeed, since $D_2$ is canonical, there is a unique such instance $\psi$.
	Since we restrict our attention to algorithms with runtime at least $m$ where $m$ is the size of the input, we get that $ |\phi|\leq time(ALG_1,\phi)$, so $p(time(ALG_1,\phi)+|\phi|)\leq p(2\cdot time(ALG_1,\phi))$.

	Therefore, let $q(x)=p(2x)$.
	Now we get that, for every input $\phi$, there exists an equivalent input $\psi$ such that $time(ALG_2,\psi)\leq q(ALG_1,\phi)$.
	Therefore, $ALG_1$ is at most as rapid as $ALG_2$ according to \cref{def:rapidity}.
	
	\textbf{Direction \emph{only if}.}
	Suppose that $ALG_1$ is at most as rapid as $ALG_2$ according to \cref{def:rapidity}.
	Then there is a polynomial $p$ such that for each input $\phi$, there is an equivalent input $\psi$ such that $time(ALG_2,\psi)\leq p(time(ALG_1,\phi))$.
	Using the monotonicity of $p$ which we assume without loss of generality, we get that $p(time(ALG_1,\phi))\leq p(time(ALG_1,\phi)+|\phi|)$.
	Lastly, since $D_2$ is canonical, any instance $\psi$ which is equivalent to $\phi$ must be the \emph{only} input instance that is equivalent to $\phi$.
	Therefore, we obtain that there exists a polynomial $p$ such that for each input $\phi$ and for all equivalent inputs $\psi$ (i.e., for the unique equivalent instance $\psi$ of $D_2$), it holds that $time(ALG_2,\psi)\leq p(time(ALG_1,\phi)+|\phi|)$.
	Therefore, $ALG_1$ is at most as rapid as $ALG_2$ according to \cref{def:rapidity-lai-rephrased}.
\end{proof}

\subsection{A Sufficient Condition for Rapidity}
\label{sec:proof-sufficient-condition}

Here, we prove \cref{thm:sufficient-condition-rapidity}, which we restate below.

\sufficientconditiontheorem*

\begin{proof} We prove the theorem for $c=1$. This can be easily extended to the case with multiple operands by treating the operands point-wise and summing their sizes.
We show that $OP(D_2)$ is at most as rapid as $OP(D_1)$, assuming that the conditions in \cref{thm:sufficient-condition-rapidity} hold. 
(Note that this swaps the roles of $D_1$ and $D_2$ relative to \cref{def:rapidity}).
In this proof, we will assume without loss of generality that all polynomials $p$ are monotone, i.e., if $x\leq y$ then $p(x)\leq p(y)$.

	We prove the theorem for a manipulation operation $OP$.
	The proof for a query operation $OP$ follows as a special case, which we treat at the end of the proof.

	Let $ALG_2$ be an $\Omega(m)$ algorithm implementing $OP(D_2)$.
	By \ref{i:rm}, we may assume without loss of generality that $ALG_2$ is runtime monotonic.
	Let $f\colon D_1\to D_2$ be the polynomial-time weakly minimizing transformation (\ref{i:d1d2}), and $g\colon D_2\to D_1$ the polynomial-time transformation in the other direction satisfying the criteria in \ref{i:d2d1}. 

	We set $ALG_1 = f \circ ALG_2 \circ g$, i.e., $ALG_1$ is as follows.
	\begin{algorithmic}[1]
		\Procedure{$ALG_1$}{$\phi$}
		\State $\psi:=f(\phi)$
		\State  $\rho := ALG_2(\psi)$	
		\State \Return $g(\rho)$							 
		\EndProcedure
	\end{algorithmic}
	$ALG_1$ is \emph{complete} (i.e., works on all inputs), since $f,g$ and $ALG_2$ are.
	The remainder of the proof shows that 
	$ALG_2$ is at most as rapid as $ALG_1$,	
	i.e., there exists a polynomial $p$ such that 
	 for all operands $\psi \in D_2$, there exists in input
	 $\phi \in D_1$ with $\ket{\phi} = \ket{\psi}$ for which $\time(ALG_1, \phi) \leq p\left(\time(ALG_2, \psi)\right)$.

Let $\psi \in D_2$.
We take $\phi \in D_1$ such that $\ket{\phi} = \ket{\psi}$ and $\sizeof{\phi} \leq s(\sizeof{\psi})$ for the polynomial $s$ ensuring the succinctness relation $D_1 \atleastassuccinctas D_2$.
Such a $\phi$ exists, because $D_1$ is more succinct than $D_2$.

It remains to show that $\exists p \colon \time(ALG_1, \phi) \leq p\left(\time(ALG_2, \psi)\right)$, where $p$ is independent of $\phi$ and $ \psi $.
	To this end, we can express the time required by $ALG_1$ by summing the runtimes of its three steps as follows.
	\begin{align}
	\label{eq:alg1-time}
	\time(ALG_1,\phi) 
	= \time(f,\phi)+\time(ALG_2,f(\phi))
	+ \time(g,ALG_2(f(\phi)))
	\end{align}
It now suffices to prove that each summand of \cref{eq:alg1-time} 
	is polynomial  in the runtime of $ALG_2(\psi)$.
 	\begin{enumerate}
	\item %
	We show $\time(f,\phi)\leq \poly(\time(ALG_2,\psi))$.
	Since $f$ runs in polynomial time in its  output (\ref{i:d1d2})
	and  $\sizeof{\phi} \leq s(\sizeof{\psi})$ (see above), 
	we have $\time(f, \phi) \leq \poly(\sizeof{f(\phi)})$.
	Let $t$ be the polynomial such that $\time(f,\phi)\leq t(|f(\phi)|)$.
	Since $f$ is weakly minimizing  (\ref{i:d1d2}), it is guaranteed that $|f(\phi)|\leq m(|\psi|)$ for some polynomial $m$.
	Lastly, by \ref{i:omega}, we have $\sizeof{\psi}  = \oh(\time(ALG^{rm}_2,\psi))$, so $|\psi|\leq k(\time(ALG_2,\psi))$ for some polynomial $k$.
	Put together, we have $\time(f,\phi)\leq t(|f(\phi)|)\leq t(m(|\psi|))\leq t(m(k(\time(ALG_2,\psi))))$, which proves the claim.
	
	\item 
	\label{item:show-fast-2}
	We show $\time(ALG_2,f(\phi))\leq \poly  (\time(ALG_2,\psi))$.
	Because $f$ is a weakly minimizing transformation (\ref{i:d1d2}), we have $\sizeof{f(\phi)}\leq s(\sizeof{\psi})$ for some $s$.
	Since $ALG_2$ is runtime monotonic (\ref{i:rm}), and because $\sizeof{f(\phi)}\leq s(\sizeof{\psi})$, we have $\time(ALG_2, f(\phi)) \leq t(ALG_2,\psi)$ for some $t$, which proves the claim.
	
	\item 
	We show $\time(g,ALG_2(f(\phi)))\leq \poly (\time(ALG_2, \psi))$.
	Since $g$ runs in time polynomial in the input (\ref{i:d2d1}); and the input to $g$ is $ALG_2(f(\phi))$, we have $\time(g,ALG_2(f(\phi)))\leq p(\sizeof{ALG_2(f(\phi))})$ for some polynomial $p$.
	Next, we have trivially $\time(ALG_2,f(\phi))\geq \sizeof{ALG_2(f(\phi))}$, since the time $ALG_2$ spends writing the output is included in the total time, thus we obtain $\time(g,ALG_2(f(\phi)))\leq p(\time(ALG_2,f(\phi)))$.
	As we have seen above in item \ref{item:show-fast-2}, $\time(ALG_2,f(\phi))\leq t(\time(ALG_2,\psi))$ for some polynomial $t$.
	Putting this together, we obtain $\time(g,ALG_2(f(\phi)))\leq p(t(\time(ALG_2,\psi)))$, which proves the claim.
	
	\end{enumerate}
	This proves the theorem for the case when $OP$ is a manipulation operation.
	
	Lastly, if $OP$ is a query operation rather than a manipulation operation, then the transformation from $D_2$ back to $D_1$ using $g$ is no longer necessary.
	This is the only change needed in $ALG_1$; in the proof above, we may use $time(g,ALG_2(f(x_1)))=0$.
	The requirement that $time(g,ALG_2(f(x_1)))\leq p(time(ALG_2,x_2))$ now holds vacuously.
\end{proof}

\subsection{Rapidity Relation between Data Structures}
\label{sec:proofs-rapidity-relations}

Here we prove the rapidity relations between data structures studied in the paper as stated in \cref{thm:rapidity}, restated below with proof.

\thmrapidity*
\begin{proof}
	The relation between \qmdd and MPS is proved in \cref{thm:mps-as-rapid-as-qmdd} as restated in \cref{sec:qmdd-to-mps}.
	Finally, \cref{sec:transformations-between-QDDs} provides the transformations between QDDs that fulfill the conditions of \cref{thm:sufficient-condition-rapidity}.
\end{proof}

\begin{figure}[h!]
\input{pics/rapidity-relations-figure}
\end{figure}

\subsection{MPS is at least as Rapid as \qmdd}
\label{sec:qmdd-to-mps}
\label{sec:mps-to-qmdd}

This appendix proves \cref{thm:mps-as-rapid-as-qmdd} from \cref{sec:transformations} by
providing transformations between MPS and \qmdd that realize the sufficient conditions of \cref{thm:sufficient-condition-rapidity}. The introduction into QDDs, given in  \cref{sec:ddprelims} is relevant here.

\mpsasrapidasqmdd*
\begin{proof}
	Let $f$ be the polynomial-time transformation from \cref{thm:trans-qmdd-to-mps}.
	Let $g$ be the weakly minimizing transformation from MPS to \qmdd of \cref{thm:trans-mps-to-qmdd}, that runs in time polynomial in the size of the input MPS and the resulting \qmdd.
These transitions satisfy requirements \ref{i:d1d2} and \ref{i:d2d1} of \cref{thm:sufficient-condition-rapidity} respectively.
Since QDDs are canonical data structures as explained in \cref{sec:preliminaries}, all algorithms are by definition runtime monotonic, as  for any state $\ket{\phi}$ there is only one structure representing it, i.e., $D^\phi$ is a singleton set.
This satisfies \ref{i:rm}. Since its premise fulfills \ref{i:omega}, the theorem follows.
\end{proof}

\begin{lemma}[\qmdd to MPS]
	\label{thm:trans-qmdd-to-mps}
	In polynomial time, a \qmdd can be converted to an MPS representing the same state.
\end{lemma}
\begin{proof}
	Consider a \qmdd with root edge $\ledge {\lambda}v$ describing a state $\ket\phi = \sum_{\vec x\in \{0,1\}^n}\alpha(\vec x)\ket{\vec x}$.
	We will construct an MPS $A$ describing the same state.
	For the purposes of this proof, we will call low edges \emph{0-edges} and high edge \emph{1-edges}. 

	First, without loss of generality, we may assume that the root edge label is $\lambda=1$.
	Namely, we may multiply the labels on the root's low and high edges with $\lambda$, and then set the root edge label to $1$; this operation preserves the state represented by the \qmdd.

\def\low{\textnormal{0}}
\def\high{\textnormal{1}}

Denote by $D_{\ell}$ the number of nodes at the $\ell$-th layer in \qmdd $v$, i.e. $D_n = 1$ (the root node $v$) and $D_{0}=1$ (the leaf \leafnode).
	Recall that the \qmdd is a directed, weighted graph whose vertices are divided into $n+1$ layers, i.e., the edges only connect nodes from consecutive layers.
	Therefore, we may speak of the $D_\ell\times D_{\ell-1}$ bipartite adjacency matrix between layer $\ell$ and layer $\ell-1$ of the diagram.
	For layer $1\leq\ell\leq n$ and $x=0,1$, let $A_\ell^x$ be the $D_\ell\times D_{\ell-1}$ bipartite adjacency matrix obtained in this way using only the low edges if $x=0$, and only the high edges if $x=1$.
	That is, assuming some order on nodes within each level, the entry of the matrix $A_{\ell}^{x}$ in row $r$ and column $c$ is defined as
  \begin{equation}
\label{eq:qmdd-to-mps}
\left(A_{\ell}^{x}\right)_{r,c} = 
    \begin{cases*}
      \textnormal{label}(e)
& if node with index $r$ in level $\ell$ has a $x$-edge $e$ to node with index $c$ in level $\ell-1$
\\
      0        & otherwise
    \end{cases*}
  \end{equation}
	We claim that the following MPS $A$ describes the same state as the \qmdd:
	\begin{equation}
	A=(A_1^{\low},A_1^{\high},\ldots, A_n^{\low},A_n^{\high})
	\end{equation}

Following the MPS definition in \cref{sec:preliminaries}, our claim is proven by showing that for \qmdd root node $v$ representing $\ket v$, we have
\begin{equation}
\label{eq:qmdd-to-mps-claim}
\langle\vec x |v\rangle = A_{n}^{x_{n}} \cdot A_{n-1}^{x_{n-1}} \cdots A_{1}^{x_{1}} \textnormal{\qquad for all $\vec x\in \{0, 1\}^n$}
\end{equation}

For an $n$-qubit \qmdd $v$, the amplitude $\langle \vec{x} | v \rangle$ for $\vec{x}\in \{0, 1\}^n$ is equal to the product of the weights found on the single path from the root node node to leaf effected by $\vec{x}$ (this path is found as follows: go down from root to leaf; at a vertex at layer $j$, choose to traverse the low edge if $x_j = 0$ and the high edge if $x_j = 1$).
We next reason that this product equals the single entry of the product $y:=A_{n}^{x_{n}} \cdot A_{n-1}^{x_{n-1}} \cdots A_{1}^{x_{1}}$ from \cref{eq:qmdd-to-mps-claim}.

We recall several useful facts from graph theory.
If $G$ ($G'$) is a weighted, directed bipartite graph on the bipartition $M \cup M''$ ($M'' \cup M'$) vertices, with weighted adjacency matrix $A_G$ ($A_{G'}$), then it is not hard to see that the element $(A_G \cdot A_{G'})_{r,c}$ is the sum, over all two-step paths $r-a-c$ starting at vertex $r \in M$ and going through vertex $a\in M''$  to vertex $c \in M'$, of products of the two weights $w_{r\rightarrow a}$ and $w_{a \rightarrow c}$.
More generally, for a sequence of weighted, directed bipartite graphs $G_j$ with vertex set $M_j \cup M_{j+1}$, the $(r,c)$-th entry of the product of adjacency matrices $A_{G_1} \cdot A_{G_2} \cdot \dots \cdot A_{G_n}$ equals $\sum_{\textnormal{paths $\pi$ from $r$ to $c$}} \prod_{\textnormal{edge} \epsilon\in \pi} \textnormal{weight}(\epsilon)$.

Now note that the matrix $y$ has dimensions $1\times 1$ (since $D_0 = D_n = 1$), corresponding to a single root and single leaf.
By the reasoning above, since $y$ is the product of all bipartite adjacency matrices of the \qmdd, the single element of this matrix is equal to the product of weights found on the single path from root to leaf as represented by $\vec{x}$.
\end{proof}

\begin{lemma}[MPS to \qmdd]
	\label{thm:trans-mps-to-qmdd}
	There is a weakly minimizing transformation from MPS to \qmdd, that runs in time polynomial in the size of the input MPS and the resulting \qmdd.
\end{lemma}
\begin{proof}
	\cref{alg:mps2qmdd} shows the algorithm which converts an MPS to a \qmdd.
	The idea is to use perform backtracking to construct the \qmdd bottom-up.
	Specifically, given an MPS $\{A_n^0, A_n^{1},\dots,A_1^0, A_{1}^1\}$ representing a state $\ket\phi = \ket 0 \ket{\phi_0} + \ket 1 \ket{\phi_1}$, the MPS  for $\ket{\phi_0}$ is easily constructed by setting the first open index to $0$ and contracting these two blocks, i.e., $A_{n-1}^{0} := A_n^0\cdot A_{n-1}^{0}$ and $A_{n-1}^1 := A_n^0\cdot A_{n-1}^1$, and similarly for $\ket{\phi_1}$.
	We then recurse, constructing MPS for states $\ket{\phi_{00}},\ket{\phi_{01}}$, etc.
	When we find a state whose \qmdd node we have already constructed, then we may simply return an edge to that \qmdd node without recursing further.
	This dynamic programming behavior is implemented through the  check at \cref{algline:mps2qmdd-dynamic-programming}.
	
Through the use of dynamic programming with the cache set $D$, it is clear that the number of recursive calls to MPS2\qmdd is bound by the number of edges in the resulting \qmdd.
Dynamic programming is implemented by checking, for each call with MPS ${M}$, whether some \qmdd node $v\in D$ already represents $\ket{M}$ up to a complex factor.
To this end, the subroutine \textsc{Equivalent}, on \cref{algline:mps2qmdd-dynamic-programming}, checks whether $\ket{M}=\lambda\cdot \ket{v}$ for some $\lambda\in\mathbb C$.
It is straightforward to see that it runs in polynomial time in the sizes of \qmdd $v$ and MPS $M$: first, it creates an MPS for the given \qmdd node $v$ using the efficient transformation in \cref{sec:qmdd-to-mps}.
Next, it computes several inner products on MPS, which can also be done in polynomial time, using the results in \cref{sec:proofs-of-sec-tractability}.
This \textsc{Equivalent} operation is called $\sizeof D$ time, which dominates the runtime of each call \textsc{MPS2\qmdd}.
Therefore the entire runtime is polynomial in the sizes of the MPS and the resulting \qmdd.

Since \qmdd is canonical, the transformation is weakly minimizing by definition.
\end{proof}
\begin{algorithm}[h!]
\begin{algorithmic}[1]
\State $D:=\{\leafnode\}$ \Comment{Initiate diagram $D$ with only a \qmdd leaf node representing 1}
\Procedure{MPS2\qmdd}{MPS $M=\{A_j^{x}\}$} \Comment{Returns a root edge $e_R$ such that $\ket{e_R} = \ket M$}
	\If{$D$ contains a node $v$ with $\ket{M} = \lambda\ket v$}
		\label{algline:mps2qmdd-dynamic-programming}
		\Return $\ledge \lambda v$\Comment{Implemented with \textsc{Equivalent}($v, M$) for all3 $v\in D$}
	\EndIf
		\State \Edge $e_0:=\textsc{MPS2\qmdd}(\{A_n^0\cdot A_{n-1}^{0},~~~ A_n^0\cdot A_{n-1}^1\}\cup \{A_{n-2}^0,A_{n-2}^1,\ldots, A_1^0,A_1^1\})$ 
		\State \Edge $e_1:=\textsc{MPS2\qmdd}(\{A_n^1\cdot A_{n-1}^{0},~~~ A_n^1\cdot A_{n-1}^1\}\cup \{A_{n-2}^0,A_{n-2}^1,\ldots, A_1^0,A_1^1\})$
	\State \Node $w:=\lnode{e_0}{}{e_1}{}$\Comment{Create new node $w$ with \textsc{MakeNode}}
	\State $D:=D\cup \{w\}$
	 \State  \Return \Edge $\ledge 1w$
\EndProcedure\vspace{1ex}
\Procedure{Equivalent}{ \qmdd \Node $v$, MPS $M=\{A_j^{x}\}$} 
	\State $V:=\qmdd\textsc{2MPS}(v)$ \Comment{Using transformation in \cref{sec:qmdd-to-mps}}
	\State $s_V:=\sqrt{|\braket{V|V}|}$ \Comment{Compute inner product}
	\State $s_M:=\sqrt{|\braket{M|M}|}$ \Comment{Compute inner product}
	\State $\lambda:=\nicefrac{1}{s_V\cdot s_M}\braket{V|M}$ \Comment{Compute inner product}
	\If{$|\lambda| = 1$}
		\Return ``$\ket{M} = \frac{s_M}{s_V}\lambda \ket{v}$''
	\Else
		~\Return ``$\ket{v}$ is not equivalent to $\ket{M}$''
	\EndIf
\EndProcedure
\end{algorithmic}
\caption{An algorithm which converts an MPS into a \qmdd.
It runs in time polynomial in $s+d$, where $s$ is the size of the \qmdd, and $d$ is the bond dimension of the MPS.
Here $D$ is the diagram representing the state.
The subroutine $\textsc{Equivalent}(v,M)$ computes whether the vectors $\ket{v},\ket{M}$ are co-linear, i.e., whether there exists $\lambda\in \mathbb C$ such that $\ket{M}=\lambda\ket{v}$.
}
\label{alg:mps2qmdd}
\end{algorithm}

\subsection{Transformations between QDDs}
\label{sec:transformations-between-QDDs}

QDDs are canonical data structures as explained in \cref{sec:preliminaries}. Therefore, (i) all algorithms are by definition runtime monotonic, as  for any state $\ket{\phi}$ there is only one structure representing it, i.e., $D^\phi$ is a singleton set; and (ii) all transformations given below are therefore weakly minimizing since they convert to a canonical data structure (namely, since they map to the unique element in $D^\phi$, in particular they map to the minimum-size element of $D^\phi$).

\subsubsection{\limdd to \qmdd}

\cref{alg:limdd2qmdd} converts a \limdd to a \qmdd in time linear in the size of the output.
The diagram is the set of edges $D$, which is initialized to contain the Leaf (i.e., the node $\ledge 11$), and is filled with the other edges during the recursive calls to \textsc{\limdd2\qmdd}.
The function $\textsc{getLexminLabel}$ is taken from Vinkhuijzen et al. \cite{vinkhuijzen2021limdd}; it returns a canonical edge label.

\begin{algorithm}[t!]
	\begin{algorithmic}[1]
	\State \qmdd $D:=\{\ledge 11\}$ \Comment{The \qmdd is initialized to contain only the Leaf}
	\Procedure{\limdd2\qmdd}{\limdd edge $\ledge {\lambda P_n\otimes P'}v$}
		\Statex Returns (a pointer to) an edge to a \qmdd node
		\If{$v$ is the Leaf node} \Return $\ledge {\lambda}{v}$
		\EndIf
		\State $R:=\textsc{getLexminLabel}(P_n\otimes P', v)$
		\If{the \textsc{Cache} contains tuple $(R,v)$} \Return $\lambda\cdot \textsc{Cache}[R,v]$
		\EndIf
		\For{$x=0,1$}
			\State \qmdd edge $r_x:=\textsc{Follow}_x(\ledge {P_n\otimes P'}{v})$
		\EndFor
		\State \qmdd edge $r:=\textsc{MakeEdge}(r_0,r_1)$
		\State $\textsc{Cache}[R,v]:=r$
		\State $D:=D\cup \{r\}$\Comment{Add the new edge to the diagram}
		\State \Return $\lambda\cdot r$
	\EndProcedure
	\end{algorithmic}
	\caption{An algorithm which converts a \limdd into an \qmdd.
	}
	\label{alg:limdd2qmdd}
\end{algorithm}

\subsubsection{\qmdd to \limdd}
\label{sec:qmdd-to-limdd}
By definition, a \qmdd can be seen as a \limdd in which every edge is labeled with a complex number and the $n$-qubit identity tensor $\mathbb I^{\otimes n}$.
Thus, a transformation does not need to do anything.
Optionally, it is possible to convert a given \limdd to one of minimum size, as described by \cite{vinkhuijzen2021limdd}.

\subsubsection{\add to \qmdd}
\label{sec:add-to-qmdd}

To convert an \add into a \qmdd, we add a Leaf node labelled with $1$; then, for each Leaf node labelled with $\lambda\ne 1$, we label each incoming edge with $\lambda$, and then reroute this edge to the (new) Leaf node labelled with $1$.
Optionally, the resulting \qmdd can be minimized to obtain the canonical instance for this state, using, e.g., techniques from \cite{miller2006qmdd,brace1990efficient}.

\subsubsection{\qmdd to \add}

\cref{alg:qmdd2add} gives a method which converts an \qmdd to an \add.
It is very similar to the ones used in the transformation \limdd to \qmdd above, .
We here check whether the diagram already contains a function which is pointwise equal to the one we are currently considering.
If so, we reuse that node; otherwise, we recurse.
\begin{algorithm}
	\begin{algorithmic}[1]
		\Procedure{\qmdd2ADD}{\qmdd edge $e=\ledge {\lambda}{v}$ on $n$ qubits}
			\If{$n=0$}
				\State $w:=\ledge{}{\lambda}$
			\ElsIf{$A$ contains a node $w$ with $\ket v=\ket w$}
				\State \Return $w$
			\Else
				\For{$x=0,1$}
					\State $w_x:=\textsc{SLDD2ADD}(\textsc{Follow}_x(e))$
				\EndFor
				\State \qmdd Node $w:=\lnode{}{w_0}{}{w_1}$
			\EndIf
			\State $A:=A\cup \{w\}$
			\State \Return $w$
		\EndProcedure
	\end{algorithmic}
	\caption{An algorithm which converts an \qmdd to an \textsc{\add}.
	Its input it an \qmdd edge $e$ representing a state $\ket e$ on $n$ qubits.
	Here the method $\textsc{Follow}_x(e)$ returns an \qmdd edge representing the state $\bra{x}\otimes \mathbb I^{\otimes n-1}\cdot \ket e$.
	It outputs an \qmdd node $w$ representing $\ket w=\ket e$.}
	\label{alg:qmdd2add}
\end{algorithm}

}

\end{document}